\documentclass{revtex4}
\usepackage{amsmath,amsthm}
\usepackage{setspace}
\usepackage{graphicx}
\newcommand{\bra}[1]{\langle#1|}
\newcommand{\ket}[1]{|#1\rangle}
\usepackage{amsfonts}
\newtheorem{theorem}{Theorem}[section]
\newtheorem{lemma}[theorem]{Lemma}
\newtheorem{proposition}[theorem]{Proposition}

\theoremstyle{remark}
\newtheorem{remark}[theorem]{Remark}

\theoremstyle{definition}
\newtheorem{definition}[theorem]{Definition}

\theoremstyle{example}
\newtheorem{example}[theorem]{Example}

\theoremstyle{notation}

\begin{document}
\title{The complete Heyting algebra of subsystems and contextuality}            
\author{A. Vourdas}
\affiliation{Department of Computing,\\
University of Bradford, \\
Bradford BD7 1DP, United Kingdom}

\begin{abstract}
The finite set of subsystems of a finite quantum system with variables in ${\mathbb Z}(n)$, is studied as a Heyting algebra.
The physical meaning of the logical connectives is discussed.
It is shown that disjunction of subsystems is more general concept than superposition. Consequently the
quantum probabilities related to projectors in the subsystems,  are incompatible
with associativity of the join in the Heyting algebra, unless if the variables belong to the same chain.
This leads to contextuality, which in the present formalism has as contexts, the chains in the Heyting algebra.
Logical Bell inequalities, which contain `Heyting factors', are discussed.
The formalism is also applied to the infinite set of all finite quantum systems, 
which is appropriately enlarged in order to become a complete Heyting algebra.
\end{abstract}

\maketitle

\section{Introduction}

It is well known that the subgroups of a group form a lattice.
This lattice contains a lot of information about the group, and in some cases (but not always) it determines the group.
Work in this area is summarized in \cite{LG1,LG2}.
Motivated by this work in group theory, we study in this paper the lattice of subsystems of a quantum system. 
The lattice formalism, and in particular the logical connectives meet, join, implication and negation, provide a language for the study of quantum systems.
We discuss the physical importance of these logical connectives, and we show that they are linked to projectors related to von Neumann measurements.

We consider a quantum system with positions in the Abelian group $G$ and momenta in its Pontryagin dual group $\widetilde G$.
We denote such a system as $\Sigma (G, \widetilde G)$.
Let $E$ be a subgroup of $G$, and $\widetilde E$ its Pontryagin dual group (which is related to $\widetilde G$
through a quotient relation, as discussed below).
We then call the system $\Sigma (E, \widetilde E)$, a subsystem of $\Sigma (G, \widetilde G)$
(or the system $\Sigma (G, \widetilde G)$ a supersystem of $\Sigma (E, \widetilde E)$).
This definition of subsystems which is based on subgroups, implies that in the semiclassical limit,
subsystems retain their identity. 

In the past few years there has been much work  on
the finite quantum systems $\Sigma ({\mathbb Z}(n), {\mathbb Z}(n))$ with positions and momenta in
${\mathbb Z}(n)$ (the ring of integers modulo $n$).
Reviews of this work have been given in \cite{FI1,FI2,FI3,FI4,FI5}.
A natural extension of this work is to consider quantum mechanics on profinite groups\cite{PRO1,PRO2}
which are at the `edge' of finite groups (in contrast to finite groups which are discrete, they are totally disconnected).
In particular, the profinite group ${\mathbb Z}_p$ ($p$-adic integers) is the inverse limit of the ${\mathbb Z}(p^n)$,
and the profinite group ${\widehat {\mathbb Z}}$ is the inverse limit of the ${\mathbb Z}(n)$.
In the former case the corresponding quantum system is $\Sigma ({\mathbb Q}_p/{\mathbb Z}_p,{\mathbb Z}_p)$
and has been studied in \cite{VOU2,VOU3} (${\mathbb Q}_p$ denotes $p$-adic numbers).
In the latter case the corresponding quantum system is 
$\Sigma ({\mathbb Q}/{\mathbb Z},{\widehat {\mathbb Z}})$ and has been studied in \cite{VOU4}
(${\mathbb Q}$ denotes rational numbers, ${\mathbb Z}$ denotes integers, and ${\widehat {\mathbb Z}}$ is defined below).
This work can be regarded as a study of `large finite quantum systems' and it factorizes
them (using the Chinese remainder theorem) as tensor products of `mathematical component systems' with dimension $p^e$ (where $p$ is a prime number).
They are fundamental building blocks of finite quantum systems, analogous to the prime numbers which are fundamental building blocks of all positive integers,
and to the $p$-Sylow groups which are the fundamental building blocks of finite and profinite Abelian groups.
Both, the number of the component systems and also the dimension of each component system can become arbitrarily large, 
but the formalism ensures that there are no divergencies (for a review see \cite{VOUREV}).

In ref.\cite{VOU1} we have studied the set of these systems $\{\Sigma ({\mathbb Z}(n), {\mathbb Z}(n))\;|\;n\in {\mathbb N}\}$,
as a directed partially ordered set with the partial order `subsystem'.
We have also added `top elements' in this set in order to make it a directed-complete partial order\cite{D1,D2,D3}. 
This includes the systems $\Sigma ({\mathbb Q}_p/{\mathbb Z}_p,{\mathbb Z}_p)$ and the system 
$\Sigma ({\mathbb Q}/{\mathbb Z},{\widehat {\mathbb Z}})$. 
In this paper we study these sets of quantum systems as distributive lattices.
Lattice theory \cite{BIR,BIR1,BIR2,BIR3} is intimately connected with logic. Special cases of distributive lattices are
Boolean algebras which are related to classical logic, and Heyting algebras which are related to intuitionistic logic (developed by Brouwer,
Heyting, Kolmogorov, etc) \cite{Hey,Hey1,Hey2}.
We explain that our lattices are complete Heyting algebras and we discuss their physical meaning, and in particular the physical importance of the non-validity of the law of the excluded middle.
The formalism uses ideas from Sylow theory for the underlying groups of positions and momenta.

Our lattice approach, provides significant insight to the problem of contextuality, from a different angle to that studied in the literature.
Since the work of Bell \cite{B1}and Kochen and Specker \cite{B2}, non-locality and contextuality have been 
studied extensively in the literature (e.g \cite{C0,C1,C2,C3,C4,C5,C6}).
Recent experimental work in this direction has been reported in \cite{E1,E2}.
Contextuality is more general than non-locality and it applies not only to multipartite systems but also to single systems.
The literature on contextuality makes clear the importance of logic in a quantum mechanical context,
and the lattice approach in the present paper is a contribution in this direction.

We explain that in our formalism, quantum probabilities associated with the projectors into subsystems, are incompatible
with associativity of the join in the Heyting algebra. This is
related to the fact that the `disjunction space' $H(m_1 \vee m_2)$, is larger than the space ${\rm span}[H(m_1)\cup H(m_2)]$ of
superpositions, and it
leads to contextuality. Contexts, in the present formalism, are chains of the Heyting algebra.
Quantum probabilities are compatible
with associativity of the join in the Heyting algebra,
only if the variables belong to the same chain.
Consequently, contexts are chains within the Heyting algebra of subsystems.
If quantum mechanics were a non-contextual theory, it would obey
the `logical Bell inequalities', studied for Boolean variables in ref\cite{A1}, and 
generalized here for Heyting variables.

In section 2 we discuss very briefly $p$-adic groups, the Sylow theory, and Heyting and Boolean algebras, in order to define the notation. In section 3 we discuss the set ${\mathbb N}_S$ of supernatural numbers as a complete Heyting algebra. 
In section 4 we define the sets ${\mathfrak Z}_S$ and $\widetilde {{\mathfrak Z}_S}$ of Abelian groups, which are used later as groups of positions and momenta of quantum systems. We show that they are complete Heyting algebras, and we discuss the meaning of the logical connectives in this formalism.

In section 5 we consider the finite set of subsystems of $\Sigma ({\mathbb Z}(n), {\mathbb Z}(n))$ and show that it is a Heyting algebra. The physical meaning of the logical connectives in this formalism is discussed in detail.
We then define contexts as chains, so that within a context,
quantum probabilities associated with the projectors into subsystems, are  compatible
with associativity of the join in the Heyting algebra.
Logical Bell inequalities are derived under the assumption that quantum mechanics is a non-contextual theory.
They are violated, and this proves that quantum mechanics is a contextual theory.
In section 6 we extend these ideas into the infinite set of subsystems of $\Sigma ({\mathbb Q}/{\mathbb Z},{\widehat {\mathbb Z}})$. 
We conclude in section 7, with a discussion of our results.

\section{Preliminaries}

\subsection{Notation}
\begin{itemize}
\item[(1)]
${\mathbb R}$ denotes the set of real numbers; ${\mathbb Q}$ the rational numbers;
${\mathbb Z}$ the integers;  ${\mathbb N}$ the natural numbers;
${\mathbb Z}_0^+={\mathbb N}\cup \{0\}$ the non-negative integers; 
and $\Pi$ the prime numbers.
\item[(2)]
$r|s$  or $r\prec s$ denotes that $r$ is a divisor of $s$.
${\mathbb D}(n)$ is the set of divisors of $n$.

A number $r$ is a Hall divisor of $n$, if it is a divisor of $n$ such that $r$ and $n/r$ are coprime.
This terminology is inspired by group theory. 
${\mathbb D}^B(n)$ is a subset of ${\mathbb D}(n)$ which contains the Hall divisors of $n$. 

${\rm GCD}(r,s)$ and ${\rm LCM}(r,s)$ are the greatest common divisor and least common multiplier correspondingly, of the integers $r,s$.

\item[(3)]
${\mathbb N}_S$ is the set of supernatural (Steinitz) numbers:
\begin{eqnarray}
{\mathbb N}_S=\left \{\prod p^{e_p}\;|\; p\in \Pi;\;\;\;\;
e_p\in {\mathbb Z}_0^+\cup \{\infty\}\right \}
\end{eqnarray}
The index $S$ in the notation indicates Steinitz or supernatural.
If all $e_p \ne \infty$ and only a finite number of them are non-zero, 
then we get the natural numbers ${\mathbb N}$. 
Sometimes, for clarity we denote as $e_p(n)$ the exponents in the factorization of $n$. 

We say that $a$ is a divisor of $b$,
when the corresponding exponents obey the relation $e_p(a)\le e_p(b)$, for all $p$.
The obvious conventions apply for inequalities that involve $\infty$. 
Also we define the generalized ${\rm LCM}(a,b)$ and ${\rm GCD}(a,b)$ as
\begin{eqnarray}\label{w2}
&&{\rm LCM}(a,b)=\prod p^{e_p};\;\;\;\;e_p=\max(e_p(a),e_p(b))\nonumber\\
&&{\rm GCD}(a,b)=\prod p^{e_p};\;\;\;\;e_p=\min(e_p(a),e_p(b)),
\end{eqnarray} 
with the obvious conventions for $\infty$.
In ${\mathbb N}_S$ we define the following elements
\begin{eqnarray}
\Omega =\prod _{p\in \Pi} p^{\infty};\;\;\;\;\;
\Omega (\pi)=\prod  _{p\in \pi}p^{\infty};\;\;\;\;\;\Omega(\pi)|\Omega
;\;\;\;\;\pi\subset \Pi
\end{eqnarray}
If $\pi$ is the empty set, then we use the convention $\Omega(\emptyset)=1$.

\item[(4)]
Let
\begin{eqnarray}
a=\prod _{p\in \Pi}p^{e_p(a)};\;\;\;\;\;b=\prod _{p\in \Pi}p^{e_p(b)}
\end{eqnarray} 
be elements of ${\mathbb N}_S$. Then
\begin{eqnarray}\label{6bn}
&&\varpi (a)=\{p\;|\;1\le e_p(a)\}=\varpi _f(a)\cup \varpi _i(a)\nonumber\\
&&\varpi _f(a)=\{p\;|\;1\le e_p(a)<\infty\}\nonumber\\
&&\varpi _i(a)=\{p\;|\;e_p(a)=\infty\}.
\end{eqnarray} 
We partition here the set $\varpi (a)$ of the primes in the factorization of $a$,
into the set $\varpi _f(a)$ of the primes with finite exponent, and the set $\varpi _i(a)$ of the primes with infinite exponent
(the indices $i$ and $f$ indicate infinite and finite exponents, correpondingly). Also
\begin{eqnarray}
&&\varpi (a>b)=\{p\;|\;e_p(a)> e_p(b)\}\nonumber\\
&&\varpi (a= b)=\{p\;|\;e_p(a)= e_p(b)\}=\Pi-[\varpi (a>b)\cup\varpi (b>a)]\nonumber\\
&&\varpi (a\ge b)=\{p\;|\;e_p(a)\ge e_p(b)\}=\varpi (a>b)\cup \varpi (a=b)=\Pi-\varpi (b>a)\nonumber\\
&&\varpi (a>b)\subseteq \varpi (a);\;\;\;\;\;\varpi (a>b)\cap \varpi _i(b)=\emptyset.
\end{eqnarray} 

\item[(5)]
If $\varpi (a)=\pi$, we say that $a$ is a $\pi$-number.
The $\Omega (\pi)$ is the maximal $\pi$-number in ${\mathbb N}_S$
(with divisibility as partial order), and 
it is a Hall divisor of $\Omega$. Also for 
$\pi \subseteq \varpi(n)$,
the $\prod _{p\in \pi}p^{e_p(n)}$ is a Hall divisor of $n$, and it is the 
maximal $\pi$-number in ${\mathbb D}(n)$.
\item[(6)]
A set $A$ viewed as a lattice (i.e., with the operations $\vee$ and $\wedge$) is denoted as $\Lambda (A)$.
Throughout the paper we have various lattices and for simplicity we use the same symbols $\prec$, $\wedge $, $\vee$ for the `partial order',
`meet' and `join', correspondingly. 
We also use the same symbols  $\cal O$ and $\cal I$ for the smallest and greatest elements in bounded lattices.
\item[(7)]
${\mathbb Z}(n)$ is the additive group of the integers modulo $n$ and it is isomorphic to
the multiplicative group of the $n$-th roots of unity
\begin{eqnarray}\label{2s}
{\cal C}(n)=\{\omega _n(\alpha)\;|\;\alpha \in {\mathbb Z}(n)\}\cong {\mathbb Z}(n);\;\;\;\;\;
\omega _n(\alpha)=\exp \left (\frac {i2\pi \alpha}{n}\right ).
\end{eqnarray}
${\mathbb Z}(1)$ (resp. ${\cal C}(1)$) contains one element which is the $0$ (resp. $1$).
If $m|n$ then  ${\mathbb Z}(m)\prec {\mathbb Z}(n)$ (here $\prec$ indicates `subgroup'). 
\item[(8)]
${\mathbb Z}^*(n)$ is the reduced system of residues modulo $n$. It contains the units 
of ${\mathbb Z}(n)$ (i.e., the elements for which ${\rm GCD}(r,n)=1$). Its cardinality is given by the Euler totient function $\varphi(n)$.

\end{itemize}

\subsection{$p$-adic groups}

${\mathbb Q}_p$ is the field of p-adic numbers and ${\mathbb Z}_p$ the ring of p-adic integers.
${\mathbb Z}_p$ is the inverse limit of ${\mathbb Z}(p^n)$,
and ${\mathbb Q}_p/{\mathbb Z}_p$ is the direct limit of ${\mathbb Z}(p^n)$
\begin{eqnarray}
\lim _{\longleftarrow}{\mathbb Z}(p^n)={\mathbb Z}_p;\;\;\;\;\;\lim _{\longrightarrow}{\mathbb Z}(p^n)={\mathbb Q}_p/{\mathbb Z}_p
\end{eqnarray}
Therefore ${\mathbb Z}_p$ is a profinite group. ${\mathbb Q}_p/{\mathbb Z}_p$ and ${\mathbb Z}_p$ are Pontryagin dual groups to each other.
Also 
\begin{eqnarray}
\widehat {\mathbb Z}=\prod _{p\in \Pi}{\mathbb Z}_p;\;\;\;\;{\mathbb Q}/{\mathbb Z}=\bigoplus _{p\in \Pi}{\mathbb Q}_p/{\mathbb Z}_p
\end{eqnarray}
$\widehat {\mathbb Z}$ is the inverse limit of ${\mathbb Z}(n)$,
and ${\mathbb Q}/{\mathbb Z}$ is the direct limit of ${\mathbb Z}(n)$:
\begin{eqnarray}
\lim _{\longleftarrow}{\mathbb Z}(n)=\widehat {\mathbb Z}
;\;\;\;\;\;\lim _{\longrightarrow}{\mathbb Z}(n)={\mathbb Q}/{\mathbb Z}.
\end{eqnarray}
Therefore ${\widehat {\mathbb Z}}$ is a profinite group. ${\mathbb Q}/{\mathbb Z}$ and ${\widehat {\mathbb Z}}$ are Pontryagin dual groups to each other. 
\begin{remark}
The direct sum $\bigoplus$ is the direct product $\prod$ with the extra condition that in its elements $(a_1,a_2,...)$
all but a finite number of the $a_i$ are equal to zero.
In Pontryagin duality, a direct product of groups becomes the direct sum of their Pontryagin dual groups. 
\end{remark}

The Pr\"ufer group ${\cal C}(p^ {\infty})$ contains 
all $p^n$-th roots of unity (for all $n \in {\mathbb Z}^+$) and it is isomorphic to ${\mathbb Q}_p/{\mathbb Z}_p$:
\begin{eqnarray}
{\cal C}(p^ {\infty})=\{\omega _{p^n}(\alpha_{p^n})|\alpha_{p^n} \in {\mathbb Z}(p^n), n \in {\mathbb Z}^+\}\cong {\mathbb Q}_p/{\mathbb Z}_p
\end{eqnarray}
Its subgroups are the multiplicative cyclic groups ${\cal C}(p^n)$ (which are isomorphic to ${\mathbb Z}(p^n)$):
\begin{eqnarray}\label{2s}
{\cal C}(p)\prec {\cal C}(p^2)\prec ...\prec {\cal C}(p^ {\infty})\cong {\mathbb Q}_p/{\mathbb Z}_p
\end{eqnarray}
These groups (for all primes $p\in \Pi$) are the building blocks of all ${\cal C}(n)$ with $n\in {\mathbb N}_S$, because of the factorization property 
\begin{eqnarray}\label{10}
n=\prod _{p\in \pi} p^{e_p}\;\;\;\;\;{\cal C}(n)=\prod _{p\in \pi} {\cal C}(p^{e_p})
\end{eqnarray}
Working with supernatural numbers, means that the product might contain an infinite number of prime numbers, and that the exponents might be equal to infinity.
For example, the Pr\"ufer group ${\cal C}(\Omega)$  is isomorphic to ${\mathbb Q}/{\mathbb Z}$:
\begin{eqnarray}
{\cal C}(\Omega)=\bigoplus _{p\in \Pi}{\cal C}(p^ {\infty})\cong \bigoplus _{p\in \Pi}{\mathbb Q}_p/{\mathbb Z}_p \cong {\mathbb Q}/{\mathbb Z}
\end{eqnarray}
Also
\begin{eqnarray}
{\cal C}[\Omega (\pi)]= \bigoplus _{p\in \pi}{\mathbb Q}_p/{\mathbb Z}_p ;\;\;\;\;\;\pi \subseteq \Pi.
\end{eqnarray}
For the empty set we use the convention ${\cal C}[\Omega (\emptyset)]={\cal C}(1)$.
The Pontryagin dual group of ${\cal C}[\Omega (\pi)]$ is the
\begin{eqnarray}\label{A19}
\widetilde {\cal C}[\Omega (\pi)]\cong \prod _{p\in \pi}{\mathbb Z}_p.
\end{eqnarray}
A summary of these groups is presented in table 1.

\subsection{Sylow theory for finite and profinite groups}

Let $p$ be a prime number and $\pi$ a set of prime numbers.
\begin{itemize}
\item
A finite or profinite group $G$ is a $\pi$-group if every prime divisor of $|G|$ (which is in general a supernatural number) belongs to $\pi$.
In the special case that $\pi$ contains only one prime $p$, the $G$ is a $p$-group.
For example, the ${\mathbb Z}_p$ is a $p$-group and the $\prod _{p\in \pi}{\mathbb Z}_p$ is a $\pi$-group.
\item
A $\pi$-Hall subgroup $G$ of a finite or profinite group $F$, 
is a $\pi$-group with coprime order $|G|$ and index
$|F:G|$.
$G$ is maximal in the sense that there is no $\pi$-subgroup of $F$ which has $G$ as a proper subgroup.
In the special case that $\pi$ contains only one prime $p$, the $G$ is a $p$-Sylow subgroup of $F$.
For example,
the ${\mathbb Z}_p$ is a $p$-Sylow subgroup of the profinite group ${\widehat {\mathbb Z}}$.
Also the $\prod _{p\in \pi}{\mathbb Z}_p$ is a $\pi$-Hall subgroup of ${\widehat {\mathbb Z}}$.

\end{itemize}

\subsection{Heyting algebras and Boolean algebras}
A special case of distributive lattices are the Heyting algebras (or Brouwer lattices).
A special case of the Heyting algebras are the Boolean algebras.
All finite distributive lattices are Heyting algebras.
There are variations in the definitions of these terms in the literature, and below we give the definitions we adopt.
\begin{definition}
\mbox{}
\begin{itemize}
\item[(1)]
A Heyting algebra is a bounded lattice in which for any 
elements $a,b$, the set of all elements $x$
which obey the relation $a\wedge x\prec b$ has a greatest element, which is denoted as $(a\Rightarrow b)$.
The $\Rightarrow $ is called implication or relative pseudocomplement  
and it is a non-commutative and non-associative binary operation. 
\item[(2)]
The $a\Rightarrow {\cal O}$ is called pseudocomplement of $a$ (negation of $a$) and is denoted as $\neg a$. 
By definition $\neg a$ is the greatest element such that $a\wedge \neg a={\cal O}$.
\item[(3)]
The equivalence connective $\Leftrightarrow$ is defined as
\begin{eqnarray}
(a\Leftrightarrow b)=(a\Rightarrow b)\wedge (b\Rightarrow a).
\end{eqnarray}
\end{itemize}
\end{definition}
The following propositions are known and we give them without proof\cite{BIR,BIR1,BIR2,BIR3}:
\begin{proposition}
\mbox{}
\begin{itemize}
\item[(1)]
In a Heyting algebra $(a\Rightarrow b)={\cal I}$ if and only if $a\prec b$.
\item[(2)]
The elements of Heyting algebras obey the following relations which involve implications
\begin{eqnarray}
&&(a\Rightarrow a)={\cal I}\\
&&a\wedge (a\Rightarrow b)=a\wedge b\\
&&b\wedge (a\Rightarrow b)=b\\
&&[a\Rightarrow(b\wedge c)]=(a\Rightarrow b)\wedge (a\Rightarrow c)
\end{eqnarray}

\item[(3)]
The elements of Heyting algebras obey the following relations which involve pseudocomplements
\begin{eqnarray}
&&a \prec \neg \neg a\label{14}\\
&&\neg a =\neg\neg \neg a\\
&&\neg(a\vee b)=(\neg a) \wedge (\neg b)\label{15a}\\
&&\neg (a\wedge b)=\neg \neg(\neg a \vee \neg b)\label{15}\\
&& \neg {\cal O}={\cal I};\;\;\;\;\;\;\neg {\cal I}={\cal O}
\end{eqnarray}
Eq.(\ref{15a}) is the first de Morgan relation, and Eq.(\ref{15}) is a weak version of the second de Morgan relation.

\end{itemize}
\end{proposition}

\paragraph*{Stone lattices:}
A Stone lattice is a Heyting algebra in which the following equivalent to each other relations, hold:
\begin{eqnarray}
&&\neg a \vee \neg \neg a={\cal I}\label{16}\\
&&\neg (a\wedge b)=\neg a \vee \neg b\label{15A}
\end{eqnarray}
Eq.(\ref{15A}) is the second de Morgan relation, and it holds in Stone lattices.
All Heyting algebras considered later are Stone lattices, and this is needed in the proof of proposition \ref{ineq} below.

\paragraph*{Boolean algebras:} 
A Boolean algebra is a Heyting algebra in which the following equivalent to each other relations, hold:
\begin{eqnarray}
&&a =\neg \neg a\label{17}\\
&&\neg a \vee a={\cal I}\label{19}
\end{eqnarray}
The second relation expresses the law of the excluded middle.
Since these two relations are equivalent, when either of them is not valid, the law of the excluded middle is not valid.
\begin{remark}
In a Boolean algebra
\begin{eqnarray}
&&(a\Rightarrow b)=(\neg a)\vee b\nonumber\\
&&(a\Leftrightarrow b)=[(\neg a)\vee b]\wedge[(\neg b)\vee a]=(a\wedge b)\vee [(\neg a)\wedge (\neg b)],
\end{eqnarray}
while in a Heyting algebra
\begin{eqnarray}
&&(\neg a)\vee b \prec (a\Rightarrow b)\nonumber\\
&&[(\neg a)\vee b]\wedge[(\neg b)\vee a]=(a\wedge b)\vee [(\neg a)\wedge (\neg b)]\prec (a\Leftrightarrow b).
\end{eqnarray}
\end{remark}
\begin{proposition}
The subset of elements of a Heyting algebra which obey the relation $a= \neg \neg a$
form a Boolean algebra with meet $\wedge$ and join $\widehat \vee$ defined as
\begin{eqnarray}
a\widehat\vee b=\neg \neg(a \vee b)
\end{eqnarray}
\end{proposition}

\section{The supernatural numbers as Heyting algebra}

\subsection{$\Lambda ({\mathbb N}_S)$ as a complete Heyting algebra}

It is known \cite{BIR} that ${\mathbb N}$  with divisibility as partial order, and
\begin{eqnarray}\label{890}
m\wedge n={\rm GCD}(m,n);\;\;\;\;\;m\vee n={\rm LCM}(m,n),
\end{eqnarray} 
is a distributive lattice which we denote as $\Lambda ({\mathbb N})$.
$\Lambda ({\mathbb N}_S)$ is a distributive lattice
with $\vee$ and $\wedge$ given by Eq.(\ref{890}), with the generalized greatest common divisor
and the generalized least common multiplier given in Eq.(\ref{w2}). 
The lattice $\Lambda ({\mathbb N}_S)$ is bounded with ${\cal O}=1$ and ${\cal I}=\Omega$.
In fact $\Lambda ({\mathbb N}_S)$ is a complete lattice.
\begin{example}\label{ex}
\begin{eqnarray}\label{79}
&&n\wedge \Omega(\pi )=n_1\Omega [\varpi _i(n)\cap \pi];\;\;\;\;\;\;n_1=\prod _{\varpi _f(n)\cap \pi}p^{e_p(n)}\nonumber\\
&&n\vee \Omega(\pi )=n_2\Omega [\varpi _i(n)\cup \pi];\;\;\;\;\;\;n_2=\prod _{\varpi _f(n)-\pi}p^{e_p(n)}.
\end{eqnarray}
\end{example}

\begin{definition}
${\mathbb N}_S^B$ (where B indicates Boolean) is the set
\begin{eqnarray}
{\mathbb N}_S^B=\{\Omega(\pi)\;|\;\pi \subseteq \Pi\}\subset {\mathbb N}_S
\end{eqnarray}
$1\in {\mathbb N}_S^B$ corresponding (by convention) to the case that $\pi$ is the empty set.
Also $\Omega \in {\mathbb N}_S^B$.
The set ${\mathbb N}_S^B$ contains all the Hall divisors of $\Omega$, i.e., all the maximal $\pi$-numbers in ${\mathbb N}_S$ (for all $\pi \subseteq \Pi$). 

\end{definition}
\begin{proposition}
\mbox{}
\begin{itemize}
\item[(1)]
$\Lambda ({\mathbb N}_S)$ is a complete Heyting algebra, but it is not a Boolean algebra.
Its sublattice  $\Lambda ({\mathbb N}_S^B)$ is a Boolean algebra, and its join $\widehat \vee$ is the same as $\vee$.
\item[(2)]
$\Lambda ({\mathbb N}_S)$ is a Stone lattice.
\end{itemize}
\end{proposition}
\begin{proof}
\mbox{}
\begin{itemize}
\item[(1)]
We show that
\begin{eqnarray}\label{100}
&&(a\Rightarrow b)=\Omega \left [\varpi (b\ge a)\right ]\prod _{p\in \varpi (a>b)}p^{e_p(b)}.
\end{eqnarray}
Therefore, for any 
elements $a,b\in {\mathbb N}_S$, the $a\Rightarrow b$ exists and consequently $\Lambda({\mathbb N}_S)$ is a Heyting algebra.
Eq.(\ref{100}) shows that if $a$ is a $\pi$-number, then $\neg a=\Omega (\Pi-\pi)$. 

Heyting algebras are Boolean algebras if all elements satisfy Eq.(\ref{19}).
We easily find elements which do not obey this relation.
For example, for any prime $p$
\begin{eqnarray}
&&\neg p=\Omega (\pi );\;\;\;\;\pi=\Pi-\{p\}\\
&&p\vee \neg p=p\Omega (\pi )\ne {\cal I}
\end{eqnarray} 
This shows that $\Lambda({\mathbb N}_S)$ is not a Boolean algebra.
We note that all the $p^n$ with $n=1,2,...$ have the same pseudocomplement.

In order to show that $\Lambda ({\mathbb N}_S^B)$ is a Boolean algebra,
it is sufficient to show that all its elements satisfy Eq.(\ref{17}) (or Eq.(\ref{19})).
It is easily seen that $\neg \Omega (\pi)=\Omega (\Pi-\pi)$ and from this follows that $\neg \neg \Omega (\pi)=\Omega (\pi)$.
In this case $\Omega (\pi _1) \vee \Omega (\pi _2)=\Omega (\pi_1\cup \pi_2)$ and it follows that the join $\widehat \vee$ is the same as $\vee$.
\item[(2)]
We have seen earlier, that if $a$ is a $\pi$-number, then $\neg a=\Omega (\Pi-\pi)$.
Consequently $\neg \neg a=\Omega (\pi)$ and 
$\neg a\vee\neg \neg a={\cal I}$. Therefore $\Lambda({\mathbb N}_S)$ is a Stone lattice.
\end{itemize}
\end{proof}

\begin{remark}
We make a number of comments which will be extended to groups and to quantum systems later.
\begin{itemize}
\item[(1)]
$\neg a$ is the maximal number in ${\mathbb N}_S$ which is coprime to $a$.
Therefore if $a$ is a $\pi$-number, the $\neg a$ is the maximal $(\Pi-\pi)$-number, and the
$\neg \neg a$ is the maximal $\pi$-number.
\item[(2)]
$(a\Rightarrow b)$ is the maximal number in ${\mathbb N}_S$ which is `partly coprime' to $a$,
in the sense that 
\begin{eqnarray}
{\rm GCD}[a,(a\Rightarrow b)]\prec b.
\end{eqnarray} 
We prove that
\begin{eqnarray}\label{100A}
&&(\neg a)\vee b=\Omega (\pi)\prod _{p\in \varpi _f(b)\cap \varpi(a)}p^{e_p(b)}\nonumber\\
&&\pi=[\Pi-\varpi (a)]\cup \varpi _i(b).
\end{eqnarray}
Comparison with Eq.(\ref{100}), using the $\pi \subseteq \varpi (b\ge a)$, shows that 
\begin{eqnarray}\label{M1}
(\neg a)\vee b \prec (a\Rightarrow b).
\end{eqnarray}
In a Boolean algebra these two quantities would have been equal to each other.

\item[(3)]
$(a\Leftrightarrow b)$ is the maximal number in ${\mathbb N}_S$ which is `partly coprime' to both $a,b$,
in the sense that 
\begin{eqnarray}
{\rm GCD}[a,(a\Leftrightarrow b)]\prec b;\;\;\;\;\;{\rm GCD}[b,(a\Leftrightarrow b)]\prec a.
\end{eqnarray} 
We prove that
\begin{eqnarray}\label{100C}
(a\Leftrightarrow b)=\Omega \left [\varpi (a= b)\right ]\prod _{p\in \varpi (a>b)}p^{e_p(b)}\prod _{p\in \varpi (b>a)}p^{e_p(a)}.
\end{eqnarray}

\end{itemize}

\end{remark}

\paragraph*{The nonvalidity of the law of the excluded middle in $\Lambda \left ({{\mathbb N}_S}\right )$ :}
If $b$ is a $\pi$-number, then 
$\neg\neg b=\Omega (\pi)$
is maximal $\pi$-number in ${\mathbb N}_S$. Therefore
$b \prec \neg \neg b$ and consequently
the law of the excluded middle is not valid.
The existence of $\pi$-numbers which are not maximal $\pi$-number in ${\mathbb N}_S$,
is intimately linked to the nonvalidity of the law of the excluded middle.
In the terminology of logic, the `true-false' terms are assigned to maximal $\pi$-numbers $\Omega (\pi)$ while the other $\pi$-numbers are the `middle'.

\subsection{The finite Heyting algebra $\Lambda [{\mathbb D}(n)]$}

$\Lambda [{\mathbb D}(n)]$ with $n\in {\mathbb N}$ is a finite distributive lattice and as such it is a complete Heyting algebra
with ${\cal O}=1$ and ${\cal I}=n$. It is also a Stone lattice.
The formulas here are similar to those for $\Lambda ({\mathbb N}_S)$, with 
\begin{eqnarray}
p^\infty\;\rightarrow \;p^{e_p(n)};\;\;\;\;\;\;\Omega (\pi)\;\rightarrow \;\prod _{p\in \pi}p^{e_p(n)};\;\;\;\;\;\;\;\Pi\;\rightarrow \;\varpi (n).
\end{eqnarray}
For example:
\begin{eqnarray}\label{ZQN}
&&\neg a=\prod _{p\in \varpi (n)-\varpi (a)}p^{e_p(n)}\nonumber\\
&&(a\Rightarrow b)=\prod _{p\in \varpi (b\ge a)}p^{e_p(n)}\prod _{p\in \varpi (a>b)}p^{e_p(b)}\nonumber\\
&&(a\Leftrightarrow b)=\prod _{p\in \varpi (a= b)}p^{e_p(n)}\prod _{p\in \varpi (a>b)}p^{e_p(b)}\prod _{p\in \varpi (b>a)}p^{e_p(a)}.
\end{eqnarray}
Its subset
\begin{eqnarray}
\Lambda [{\mathbb D}^B(n)]=\left \{\prod _{p\in \pi}p^{e_p(n)}\;|\;\pi \subseteq \varpi(n)\right \},
\end{eqnarray}
is a Boolean algebra.
The set ${\mathbb D}^B(n)$ contains all maximal $\pi$-numbers in ${\mathbb D}(n)$ (for all $\pi \subseteq \varpi(n)$) and also $1$. 
If all exponents $e_p(n)=1$ then ${\mathbb D}^B(n)={\mathbb D}(n)$.

\section{The complete Heyting algebras $\Lambda ({\mathfrak Z}_S)$ and $\Lambda (\widetilde {{\mathfrak Z}_S})$}

In quantum mechanics we have an Abelian group $G$ where the variable `position' takes values,
and its Pontryagin dual group $\widetilde G$ where the variable `momentum' takes values.
We consider the following sets of `groups of positions' and `groups of momenta' 
\begin{eqnarray}
&&{\mathfrak Z}_S=\{{\cal C}(m)\;|\;m\in {\mathbb N}_S\}\nonumber\\
&&\widetilde {{\mathfrak Z}_S}=\{\widetilde {\cal C}(m)\;|\;m\in {\mathbb N}_S\}.
\end{eqnarray}
In this section we study these sets as lattices and in particular as complete Heyting algebras which are isomorphic to each other.
All elements of $\widetilde {{\mathfrak Z}_S}$ are finite or profinite groups, and we can apply Sylow theory.
In particular the 
\begin{eqnarray}
\widetilde {\cal C}(p)\cong {\mathbb Z}(p),\;\;\;\widetilde {\cal C}(p^2)\cong {\mathbb Z}(p^2),...,\widetilde {\cal C}(p^\infty)\cong {\mathbb Z}_p
\end{eqnarray}
are $p$-groups, and the ${\mathbb Z}_p$ is a $p$-Sylow group.
More generally if $m$ is a $\pi$-number then the
\begin{eqnarray}
\widetilde {\cal C}(m)\cong \prod _{p\in \pi}{\mathbb Z}(p^{e_p}),\;\;\;\widetilde {\cal C}(m^2)\cong \prod _{p\in \pi}{\mathbb Z}(p^{2e_p}),...,
\widetilde {\cal C}(m^\infty)\cong \prod _{p\in \pi}{\mathbb Z}_p;\;\;\;\;
m=\prod _{p\in \pi}p^{e_p}
\end{eqnarray}
are $\pi$-groups, and the $\prod _{p\in \pi}{\mathbb Z}_p$ is a $\pi$-Hall group.

We note that the ${\mathfrak Z}_S$ contains elements 
which are not finite or profinite groups (e.g. the ${\cal C}(p^\infty)\cong {\mathbb Q}_p/{\mathbb Z}_p$).

\subsection{The complete Heyting algebra $\Lambda ({\mathfrak Z}_S)$}

The set ${\mathfrak Z}_S$ is a directed partially ordered set with `subgroup' as partial order.
Actually it is a distributive lattice with
\begin{eqnarray}\label{T5}
{\cal C}(m)\wedge {\cal C}(n)={\cal C}(m\wedge n);\;\;\;\;\;\;{\cal C}(m)\vee {\cal C}(n)={\cal C}(m\vee n),
\end{eqnarray}
which we denote as $\Lambda ({\mathfrak Z}_S)$.
Here ${\cal C}(m\wedge n)$ is the largest group in the set ${\mathfrak Z}_S$, which is a subgroup of both ${\cal C}(m)$ and ${\cal C}(n)$,
and ${\cal C}(m\vee n)$ is the smallest group in the set ${\mathfrak Z}_S$, which has both ${\cal C}(m)$ and ${\cal C}(n)$ as subgroups.

$\Lambda ({\mathfrak Z}_S)$ is a bounded lattice with 
\begin{eqnarray}\label{bt}
{\cal O}={\mathbb Z}(1);\;\;\;\;\;{\cal I}={\cal C}(\Omega)\cong {\mathbb Q}/{\mathbb Z}, 
\end{eqnarray}
and it contains the subgroups of ${\cal C}(\Omega)\cong {\mathbb Q}/{\mathbb Z}$.

\begin{example}\label{ex2}
Taking into account example \ref{ex}, we find that 
for $n\in {\mathbb N}_S$
\begin{eqnarray}
{\mathbb Z}(n)\wedge \bigoplus _{p\in \pi}({\mathbb Q}_p/{\mathbb Z}_p)={\mathbb Z}(n_1)\bigoplus _{p\in \varpi _i(n)\cap \pi}({\mathbb Q}_p/{\mathbb Z}_p)\nonumber\\
{\mathbb Z}(n)\vee \bigoplus _{p\in \pi}({\mathbb Q}_p/{\mathbb Z}_p)
=\mathbb Z(n_2)\bigoplus _{p\in \varpi _i(n)\cup \pi}({\mathbb Q}_p/{\mathbb Z}_p),
\end{eqnarray}
where $n_1$, $n_2$ have been defined in Eq.(\ref{79}).
\end{example}

$\Lambda ({\mathfrak Z}_S)$ is isomorphic to $\Lambda ({\mathbb N}_S)$ and therefore:
\begin{proposition}
$\Lambda ({\mathfrak Z}_S)$ is a complete Heyting algebra (and a Stone lattice).
\end{proposition}

\paragraph*{The nonvalidity of the law of the excluded middle in $\Lambda \left ({{\mathfrak Z}_S}\right )$ :}
$\neg {\cal C}(n)$ is the maximal group in ${\mathfrak Z}_S$ such that $\neg {\cal C}(n) \wedge {\cal C}(n)={\cal C}(1)$. 
If $n$ is a $\pi$-number, then 
\begin{eqnarray}\label{D1}
\neg {\cal C}(n)=\bigoplus _{p\in \Pi-\pi}({\mathbb Q}_p/{\mathbb Z}_p)
;\;\;\;\;\;
\neg\neg {\cal C}(n)=\bigoplus_{p\in \pi}({\mathbb Q}_p/{\mathbb Z}_p)
\end{eqnarray}
 Therefore
${\cal C}(n) \prec \neg\neg {\cal C}(n)$ and consequently
the law of the excluded middle is not valid.
Here the `middle' are the ${\cal C}(n)$, where $n$ is a non-maximal $\pi$-number.
If we delete them from the set ${\mathfrak Z}_S$, we get the set
\begin{eqnarray}\label{129}
{\mathfrak Z}_S^B=\{{\cal C}[\Omega (\pi)]\;|\;\pi \subseteq \Pi\};\;\;\;\;\;{\cal C}[\Omega (\pi)]
=\bigoplus _{p\in \pi}({\mathbb Q}_p/{\mathbb Z}_p),
\end{eqnarray}
which is a Boolean algebra.

We note that $[\neg {\cal C}(n)]\vee [\neg\neg {\cal C}(n)]={\cal I}$,
as it should be  in a Stone lattice.

\paragraph*{The implication in $\Lambda ({\mathfrak Z}_S)$:}

Taking into account Eqs.(\ref{100}),(\ref{100A}), we show that
\begin{eqnarray}\label{100B}
[{\cal C}(m)\Rightarrow {\cal C}(n)]={\cal C}(m\Rightarrow n)=\left [\bigoplus _{p\in \varpi (n\ge m)}({\mathbb Q}_p/{\mathbb Z}_p)\right ]\times
{\cal C}\left (\prod _{p\in \varpi (m>n)}p^{e_p(n)}\right),
\end{eqnarray}
and that
\begin{eqnarray}\label{101B}
&&(\neg {\cal C}(m))\vee {\cal C}(n)={\cal C}[(\neg m)\vee n]=
\left [\bigoplus _{p\in \pi}({\mathbb Q}_p/{\mathbb Z}_p)\right ]\times {\cal C}\left (\prod _{p\in \varpi _f(n)\cap \varpi (m)}
p^{e_p(n)}\right )\nonumber\\
&&\pi=[\Pi-\varpi (m)]\cup \varpi _i(n).
\end{eqnarray}
Therefore
\begin{eqnarray}\label{M2}
(\neg {\cal C}(m))\vee {\cal C}(n)\prec [{\cal C}(m)\Rightarrow {\cal C}(n)].
\end{eqnarray}

\subsubsection{The finite Heyting algebra $\Lambda [{\mathfrak Z}(n)]$}
We also consider the set 
\begin{eqnarray}
{\mathfrak Z}(n)=\{{\mathbb Z}(m)\;|\;m\in {\mathbb D}(n)\};\;\;\;\;\;n\in{\mathbb N},
\end{eqnarray}
which contains the subgroups of ${\mathbb Z}(n)$.
It is a complete Heyting algebra with 
\begin{eqnarray}
{\cal O}={\mathbb Z}(1);\;\;\;\;\;{\cal I}={\mathbb Z}(n). 
\end{eqnarray}
The logical connectives can be defined in analogous way to Eq.(\ref{ZQN}).
Its subset
\begin{eqnarray}
\Lambda [{\mathfrak Z}^B(n)]=\left \{{\mathbb Z}\left (\prod _{p\in \pi}p^{e_p(n)}\right )\;|\;\pi \subseteq \varpi(n)\right \},
\end{eqnarray}
contains the $\pi$-Hall subgroups of ${\mathbb Z}(n)$, and it 
is a Boolean algebra.

\subsection{The complete Heyting algebra $\Lambda (\widetilde {{\mathfrak Z}_S})$}

We have defined binary operations in ${\mathfrak Z}_S$, and through Pontryagin duality we can define the corresponding operations
in $\widetilde {{\mathfrak Z}_S}$:
\begin{eqnarray}
&&{\widetilde {\cal C}}(n)\wedge {\widetilde {\cal C}}(m)=\widetilde {{\cal C}(n)\wedge {\cal C}(m)}={\widetilde {\cal C}}(n\wedge m)\nonumber\\
&&{\widetilde {\cal C}}(n)\vee {\widetilde {\cal C}}(m)=\widetilde {{\cal C}(n)\vee {\cal C}(m)}={\widetilde {\cal C}}(n\vee m)
\end{eqnarray}
The partial order in $\Lambda (\widetilde {{\mathfrak Z}_S})$ is different from the one in $\Lambda ({{\mathfrak Z}_S})$.
The partial order in $\Lambda ({{\mathfrak Z}_S})$ is `subgroup', and this implies that there is a
quotient relation between their Pontryagin dual groups in $\Lambda (\widetilde {{\mathfrak Z}_S})$, which involves the annihilators of 
the groups:
\begin{eqnarray}
{\mathbb Z}(p^e)\cong \widehat {\mathbb Z}/{\rm Ann}_{\widehat {\mathbb Z}}{\mathbb Z}(p^e);\;\;\;\;\;
{\mathbb Z}_p\cong \widehat {\mathbb Z}/{\rm Ann}_{\widehat {\mathbb Z}}({\mathbb Q}_p/{\mathbb Z}_p)
\end{eqnarray}
Details in a general context are discussed in \cite{PO}, and in the present context in \cite{VOUREV}.  

$\Lambda (\widetilde {{\mathfrak Z}_S})$  is isomorphic to $\Lambda ({\mathfrak Z}_S)$
and therefore it is a complete Heyting algebra (and a Stone lattice), with
\begin{eqnarray}\label{300}
{\cal O}={\mathbb Z}(1);\;\;\;\;\;{\cal I}=\widehat {\mathbb Z}. 
\end{eqnarray}
Its elements are $\pi$-groups and $p$-groups (see Eq.(\ref{A19})) which in general are not subgroups of $\widehat {\mathbb Z}$,
although in some cases (as in proposition \ref{bb} below) they are subgroups of $\widehat {\mathbb Z}$.

\paragraph*{The implication in $\Lambda (\widetilde {{\mathfrak Z}_S})$:}
We show that
\begin{eqnarray}\label{100BB}
[{\widetilde {\cal C}}(m)\Rightarrow {\widetilde {\cal C}}(n)]=[\widetilde {{\cal C}(m)\Rightarrow {\cal C}(n)}]={\widetilde {\cal C}}(m\Rightarrow n)
=\left [\prod _{p\in \varpi (n\ge m)}{\mathbb Z}_p\right ]
{\widetilde {\cal C}}\left (\prod _{p\in \varpi (m>n)}p^{e_p(n)}\right),
\end{eqnarray}
and also that
\begin{eqnarray}\label{303}
&&(\neg {\widetilde {\cal C}}(m))\vee {\widetilde {\cal C}}(n)={\widetilde {\cal C}}[(\neg m)\vee n]=
\left [\prod _{p\in \pi}{\mathbb Z}_p\right ]\times {\cal C}\left (\prod _{p\in \varpi _f(n)\cap \varpi (m)}
p^{e_p(n)}\right )\nonumber\\
&&\pi=[\Pi-\varpi (m)]\cup \varpi _i(n).
\end{eqnarray}
Since this is a Heyting algebra, these two quantities are different.

\paragraph*{Negation:}
If $n$ is a $\pi$-number, then 
\begin{eqnarray}\label{D2}
\neg {\widetilde {\cal C}}(n)=\prod _{p\in \Pi-\pi}{\mathbb Z}_p
;\;\;\;\;\;
\neg\neg {\widetilde {\cal C}}(n)=\prod _{p\in \pi}{\mathbb Z}_p.
\end{eqnarray}
We express this in the following proposition:
\begin{proposition}\label{bb}
Let $G$ be a $\pi$-group which is an element of $\Lambda (\widetilde {{\mathfrak Z}_S})$.
\begin{itemize}
\item[(1)]
$\neg G=\prod _{p\in \Pi-\pi}{\mathbb Z}_p$ is a ($\Pi-\pi$)-Hall subgroup of the profinite group ${\widehat {\mathbb Z}}$.
\item[(2)]
$\neg \neg G =\prod _{p\in \pi}{\mathbb Z}_p$ is a $\pi$-Hall subgroup of the profinite group ${\widehat {\mathbb Z}}$.
\end{itemize}
\end{proposition}

\paragraph*{The Boolean algebra $\Lambda \left (\widetilde {{\mathfrak Z}_S^B}\right )$:}
The subset of $\widetilde{\mathfrak Z}_S$ given by
\begin{eqnarray}\label{301}
\widetilde{{\mathfrak Z}_S^B}=\{\prod _{p\in \pi}{\mathbb Z}_p\;|\;\pi \subset \Pi\},
\end{eqnarray}
with $\wedge$ and $\vee$ operations,
is the Boolean algebra $\Lambda \left (\widetilde {{\mathfrak Z}_S^B}\right )$,
and it contains $\pi$-Hall subgroups of ${\widehat {\mathbb Z}}$
(and $p$-Sylow subgroups of ${\widehat {\mathbb Z}}$).

\paragraph*{The finite Heyting algebra $\Lambda (\widetilde {{\mathfrak Z}(n))}$:}
Since the ${\mathbb Z}(n)$ with $n\in {\mathbb N}$ are Pontryagin self-dual groups, the
$\Lambda (\widetilde {{\mathfrak Z}(n))}$ is identical to $\Lambda [{{\mathfrak Z}(n)]}$.

\section{The Heyting algebra $\Lambda[{\bf \Sigma}(n)]$ of subsystems}

\subsection{Subsystems}
$\Sigma [{\mathbb Z}(n), {\mathbb Z}(n)]$ is a quantum system with positions and momenta in ${\mathbb Z}(n)$.
The Hilbert space $H(n)$ for this system is $n$-dimensional, and 
$|X_n;r\rangle$ where $r\in {\mathbb Z}(n)$, is an orthonormal basis that we call `basis of position states'
(the $X_n$ in this notation is not a variable, but it simply indicates that they are position states).
Through a Fourier transform we get another orthonormal basis that we call momentum states:
\begin{equation}
|{P_n};r\rangle=F_n|{X_n};r\rangle;\;\;\;\;
F_n=n^{-1/2}\sum _{r,s}\omega _n(rs)\ket{X_n;r}\bra{X_n;s}.
\end{equation}
\begin{remark}
\mbox{}
\begin{itemize}
\item[(1)]
The system $\Sigma [{\mathbb Z}(1),{\mathbb Z}(1)]$ is 
physically trivial, as it has one-dimensional Hilbert space $H(1)$ which consists of the `vacuum' state $\ket{X_1;0}=\ket{P_1;0}$.
\item[(2)]
A system $\Sigma (E, \widetilde E)$ is a subsystem of $\Sigma (G, \widetilde G)$ if $E$ is a subgroup of G (in which case the
$\widetilde E$ is related to $\widetilde G$ through a quotient relation that involves their annihilators).
For $m|n$ the $\Sigma [{\mathbb Z}(m),{\mathbb Z}(m)]$ is a subsystem of $\Sigma [{\mathbb Z}(n),{\mathbb Z}(n)]$
(which we denote as $\Sigma [{\mathbb Z}(m),{\mathbb Z}(m)]\prec \Sigma [{\mathbb Z}(n),{\mathbb Z}(n)]$), and
the space $H(m)$ is a subspace of $H(n)$ (which we denote as $H(m)\prec H(n)$).
\end{itemize}
\end{remark}
\begin{definition}[inspired by the Sylow theory for groups]
\mbox{}
\begin{itemize}
\item[(1)]
$\Sigma (G, \widetilde G)$ is a $\pi$-system if $\widetilde G$ is a $\pi$-group.
If $\pi$ contains only one prime $p$, the $\Sigma (G, \widetilde G)$ is a $p$-system.
\item[(2)]
$\Sigma (E, \widetilde E)$ is a $\pi$-Hall subsystem of $\Sigma (G, \widetilde G)$,
if $\widetilde E$ is a $\pi$-Hall subgroup of $\widetilde G$.
In this case $\Sigma (E, \widetilde E)$ is maximal in the sense that there is no $\pi$-subsystem of $\Sigma (G, \widetilde G)$
which has $\Sigma (E, \widetilde E)$ as a proper subsystem.
If $\pi$ contains only one prime $p$, the $\Sigma (E, \widetilde E)$ is a $p$-Sylow subsystem of 
$\Sigma (E, \widetilde E)$.
\end{itemize}
\end{definition}
Let ${\bf \Sigma}(n)$ be the set of subsystems of $\Sigma [{\mathbb Z}(n),{\mathbb Z}(n)]$ and ${\bf H}(n)$ the set of their Hilbert spaces:
\begin{eqnarray}
&&{\bf \Sigma} (n)=\{\Sigma [{\mathbb Z}(m),{\mathbb Z}(m)]\;|\;m\in {\mathbb D}(n)\}\nonumber\\
&&{\bf H}(n) =\{H(m)\;|\;m\in {\mathbb D}(n)\}.
\end{eqnarray}
${\bf \Sigma}(n)$ is a partially ordered set with partial order `subsystem'.
${\bf H}(n)$ is a partially ordered set with partial order `subspace'.
The set ${\bf H}(n)$ does not contain all the subspaces of a Hilbert space $H(n)$,
but only the ones that correspond to subsystems of $\Sigma [{\mathbb Z}(n),{\mathbb Z}(n)]$.
The concept `Hilbert space of a subsystem of $\Sigma [{\mathbb Z}(n),{\mathbb Z}(n)]$', is stronger than the concept of subspace of $H(n)$. 
Below we study these sets as Heyting algebras.

\paragraph*{Embedings and their compatibility:}
In ref\cite{VOU1}, we have studied ${\bf \Sigma}(n) -\{\Sigma [{\mathbb Z}(1),{\mathbb Z}(1)]\}$ as a partially ordered set with the partial order `subsystem' 
(it is not a lattice).
We have shown that for $m|k$, the $\Sigma [{\mathbb Z}(m),{\mathbb Z}(m)]$ is a subsystem of 
$\Sigma [{\mathbb Z}(k),{\mathbb Z}(k)]$ and that 
there are embeddings of various attributes related to $\Sigma [{\mathbb Z}(m),{\mathbb Z}(m)]$
(quantum states, density matrices, measurements, operators, etc)
into their counterparts in $\Sigma [{\mathbb Z}(k),{\mathbb Z}(k)]$, which are compatible with each other and preserve the structure.
For example, all states in $H(m)$ are mapped into states in $H(k)$ with the linear map
\begin{eqnarray}\label{56}
{\cal A}_{mk}:\;\;\sum _{r=0}^{m-1}a_r\ket{X_m;r}\;\;\rightarrow\;\;\sum _{r=0}^{m-1}a_r\ket{X_k;dr};\;\;\;\;d=\frac{k}{m};\;\;\;\;m|k.
\end{eqnarray}
The same map can be written in terms of momentum states as
\begin{eqnarray}\label{56a}
&&{\cal A}_{mk}:\;\;\sum _{r=0}^{m-1}b_r\ket{P_m;r}\;\;\rightarrow\;\;\sum _{s=0}^{k-1}c_s\ket{P_k;s};\;\;\;\;m|k.\nonumber\\
&&s=r\;({\rm mod}\;m)\;\rightarrow\;c_s=d^{-1/2}b_r;\;\;\;\;d=\frac{k}{m}
\end{eqnarray}
We call ${\cal A}_{mk}[H(m)]$ the image of the function ${\cal A}_{mk}$.
The space ${\cal A}_{mk}[H(m)]$ is a subspace of $H(k)$ which is isomorphic to $H(m)$.
We have shown that these maps are compatible in the sense that 
\begin{eqnarray}\label{comp}
m|s|k\;\;\rightarrow\;\;{\cal A}_{sk}\circ {\cal A}_{ms}={\cal A}_{mk}.
\end{eqnarray}
The map ${\cal A}_{mk}$ induces the following map for density matrices 
\begin{eqnarray}\label{56b}
{\cal A}'_{mk}:\;\;\rho _m=\sum _{r,s=0}^{m-1}a_{rs}\ket{X_m;r}\bra{X_m;s}\;\;\rightarrow\;\;\rho _k=\sum _{r,s=0}^{m-1}a_{rs}\ket{X_k;dr}\bra{X_k;ds};\;\;\;\;d=\frac{k}{m};\;\;\;\;m|k.
\end{eqnarray}

In this paper we include the system $\Sigma [{\mathbb Z}(1),{\mathbb Z}(1)]$ into ${\bf \Sigma}(n)$ so that it becomes a lattice.
\begin{definition}
$h(n)$ is the following $\varphi(n)$-dimensional subspace of $H(n)$:
\begin{eqnarray}\label{2}
h(n)&=&{\rm span}\{\ket{X_n;s}\;|\;s\in {\mathbb Z}^*(n)\}.
\end{eqnarray}
\end{definition}
\begin{lemma}\label{790}
For $m|n$, let ${\cal A}_{mn}[h(m)]$ be the embedding of the space $h(m)$ into $H(n)$, given in Eq.(\ref{56}). 
If $m,k\in {\mathbb D}(n)$, then 
the spaces ${\cal A}_{mn}[h(m)]$ and ${\cal A}_{kn}[h(k)]$ are orthogonal to each other, and
\begin{eqnarray}\label{606}
H(n)=\bigoplus _{m|n}{\cal A}_{mn}[h(m)].
\end{eqnarray}
\end{lemma}
\begin{proof}
We first point out that if $m,k$ are two different divisors of $n$, then 
\begin{eqnarray}
{\cal A}_{mn}[h(m)]\cap {\cal A}_{kn}[h(k)]={\bf O}.
\end{eqnarray}
${\bf O}$
is the zero-dimensional space 
which contains only the zero vector. Indeed, if ${\cal A}_{mn}[\ket{X_m;r}]={\cal A}_{kn}[\ket{X_k;s}]$
(where $r\in {\mathbb Z}^*(m)$ and $s \in {\mathbb Z}^*(k)$), then
\begin{eqnarray}
\frac{nr}{m}=\frac{ns}{k};\;\;\;\;\;{\rm GCD}(r,m)={\rm GCD}(s,k)=1,
\end{eqnarray}
which is an incompatible set of equations.

We next consider the position state $\ket{X_n;r}$ in $H(n)$, and let ${\rm GCD}(r,n)=a$.
Then
\begin{eqnarray}
\ket{X_n;r}={\cal A}_{mn}[\ket{X_m;r'}];\;\;\;\;\;m=\frac{n}{a};\;\;\;\;r'=\frac{r}{a};\;\;\;\;{\rm GCD}(r',m)=1,
\end{eqnarray}
where $\ket{X_m;r'}$ is a state in $h(m)$. Consequently, an arbitrary state $\sum a_r \ket{X_n;r}$ in $H(n)$ can be written as
sum of states in various ${\cal A}_{mn}[h(m)]$. This completes the proof. We note that the dimension of $h(n)$ is $\varphi(n)$, and the known relation
\begin{eqnarray}
\sum _{m|n}\varphi (m)=n
\end{eqnarray}
is consistent with Eq.(\ref{606}).
\end{proof}
For later use we define the space ${\widetilde H}(n)$ by excluding the lowest state from $H(n)$:
\begin{eqnarray}\label{cde}
{\widetilde H}(n)=\bigoplus _{m}{\cal A}_{mn}[h(m)];\;\;\;\;\;m\in {\mathbb D}(n)-\{1\}.
\end{eqnarray}

\subsection{The Heyting algebra of subsystems $\Lambda[{\bf \Sigma}(n)]$ and the physical meaning of the logical operations}

The set ${\bf \Sigma}(n) $  with 
\begin{eqnarray}\label{123}
&&{\Sigma}[{\mathbb Z}(m),{\mathbb Z}(m)]\wedge {\Sigma}[{\mathbb Z}(k),{\mathbb Z}(k)]=
{\Sigma}[{\mathbb Z}(m\wedge k),{\mathbb Z}(m \wedge k)]\nonumber\\
&&\Sigma [{\mathbb Z}(m),{\mathbb Z}(m)]\vee \Sigma [{\mathbb Z}(k),{\mathbb Z}(k)]=\Sigma [{\mathbb Z}(m\vee k),{\mathbb Z}(m\vee k)]
\end{eqnarray}
where $m,k\in {\mathbb D}(n)$, is a finite distributive lattice which we denote as $\Lambda [{\bf \Sigma}(n)] $. Therefore it is a Heyting algebra with
\begin{eqnarray}
{\cal O}=\Sigma [{\mathbb Z}(1),{\mathbb Z}(1)];\;\;\;\;\;{\cal I}=\Sigma [{\mathbb Z}(n),{\mathbb Z}(n)], 
\end{eqnarray}
and it is
isomorphic to $\Lambda ({\mathbb D}(n))$.
In analogous way we define the logical operations in ${\bf H}(n)$, which 
is a Heyting algebra isomorphic to 
$\Lambda [{\mathbb D}(n)]$ and $\Lambda [{\bf \Sigma}(n)]$, and we denote it as $\Lambda [{\bf H}(n)]$.

\paragraph*{Physical meaning of the meet:}
The ${\Sigma}[{\mathbb Z}(m),{\mathbb Z}(m)]\wedge {\Sigma}[{\mathbb Z}(k),{\mathbb Z}(k)]$ 
(where $m,k\in {\mathbb D}(n)$) is the largest system in the set ${\bf \Sigma}(n)$ 
which is a subsystem of both 
${\Sigma}[{\mathbb Z}(m),{\mathbb Z}(m)]$ and ${\Sigma}[{\mathbb Z}(k),{\mathbb Z}(k)]$.  
It quantifies the commonality between the systems ${\Sigma}[{\mathbb Z}(m),{\mathbb Z}(m)]$ and ${\Sigma}[{\mathbb Z}(k),{\mathbb Z}(k)]$.
Its Hilbert space $H(m\wedge k)$ is the intersection $H(m)\cap H(k)$.

Using terminology inspired by number theory, we say that 
${\Sigma}[{\mathbb Z}(m),{\mathbb Z}(m)]$ and ${\Sigma}[{\mathbb Z}(k),{\mathbb Z}(k)]$
are coprime systems, when ${\Sigma}[{\mathbb Z}(m),{\mathbb Z}(m)]\wedge {\Sigma}[{\mathbb Z}(k),{\mathbb Z}(k)]={\Sigma}[{\mathbb Z}(1),{\mathbb Z}(1)]$
(i.e., when $m,k$ are coprime numbers). Physically these systems contain complementary information, in the sense that they share only the `lowest state'
$\ket{X_1;0}$.

\paragraph*{Physical meaning of the join:}
The $\Sigma [{\mathbb Z}(m),{\mathbb Z}(m)]\vee \Sigma [{\mathbb Z}(k),{\mathbb Z}(k)]$ (where $m,k\in {\mathbb D}(n)$) 
is the smallest system in the set 
${\bf \Sigma}(n)$ which contains both
${\Sigma}[{\mathbb Z}(m),{\mathbb Z}(m)]$ and ${\Sigma}[{\mathbb Z}(k),{\mathbb Z}(k)]$ as subsystems. 
Its Hilbert space $H(m\vee k)$ is discussed in detail below.
We will see that $H(m\vee k)$ is a larger space than ${\rm span}[H(m)\cup H(k)]$, and this has important implications.

The logical operations 
\begin{eqnarray}\label{123}
&&\neg {\Sigma}[{\mathbb Z}(m),{\mathbb Z}(m)]={\Sigma}[{\mathbb Z}(\neg m),{\mathbb Z}(\neg m)]\nonumber\\
&&\left ({\Sigma}[{\mathbb Z}(m),{\mathbb Z}(m)]\Rightarrow {\Sigma}[{\mathbb Z}(k),{\mathbb Z}(k)]\right )=
{\Sigma}[{\mathbb Z}(m\Rightarrow k),{\mathbb Z}(m \Rightarrow k)]\nonumber\\
&&\left ({\Sigma}[{\mathbb Z}(m),{\mathbb Z}(m)]\Leftrightarrow {\Sigma}[{\mathbb Z}(k),{\mathbb Z}(k)]\right )=
{\Sigma}[{\mathbb Z}(m\Leftrightarrow k),{\mathbb Z}(m \Leftrightarrow k)]
\end{eqnarray}
can be calculated using Eq.(\ref{ZQN}).

\paragraph*{Physical meaning of the negation:}
$\neg \Sigma[{\mathbb Z}(m),{\mathbb Z}(m)]$ is the maximal system in ${\bf \Sigma}(n)$ which is coprime to $\Sigma[{\mathbb Z}(m),{\mathbb Z}(m)]$.
If $\Sigma[{\mathbb Z}(m),{\mathbb Z}(m)]$ is a $\pi$-system (where $\pi\subset \varpi (n)$) then 
$\neg \Sigma[{\mathbb Z}(m),{\mathbb Z}(m)]$ is a $(\varpi (n)-\pi)$-Hall subsystem of $\Sigma [{\mathbb Z}(n),{\mathbb Z}(n)]$, i.e., it
is the maximal $(\varpi (n)-\pi)$-system in ${\bf \Sigma}(n)$ which contains
complementary information to $\Sigma [{\mathbb Z}(m),{\mathbb Z}(m)]$.
Also
$\neg \neg \Sigma[{\mathbb Z}(m),{\mathbb Z}(m)]$ is a $\pi$-Hall subsystem of $\Sigma [{\mathbb Z}(n),{\mathbb Z}(n)]$, and consequently
$\Sigma[{\mathbb Z}(m),{\mathbb Z}(m)]$ is a subsystem of $\neg \neg \Sigma[{\mathbb Z}(m),{\mathbb Z}(m)]$
and the law of the excluded middle is not valid.
The `middle' here is $\pi$-subsystems which are not $\pi$-Hall subsystems of $\Sigma [{\mathbb Z}(n),{\mathbb Z}(n)]$.
If we delete them from the set ${\bf \Sigma}(n)$ we get 
its subset ${\bf \Sigma}^B(n)$ which is a Boolean algebra (it contains the systems ${\Sigma}[{\mathbb Z}(m),{\mathbb Z}(m)]$
with $m\in {\mathbb D}^B(n)$).

We note that $\neg \Sigma[{\mathbb Z}(m),{\mathbb Z}(m)]\vee \neg \neg \Sigma[{\mathbb Z}(m),{\mathbb Z}(m)]={\cal I}$,
i.e, $\Lambda[{\bf \Sigma}(n)]$ is a Stone lattice.

\paragraph*{Physical meaning of the implication:}
The implication $[\Sigma [{\mathbb Z}(m),{\mathbb Z}(m)]\Rightarrow \Sigma [{\mathbb Z}(k),{\mathbb Z}(k)]]$ 
is the maximal system in ${\bf \Sigma}(n)$ which is `partly coprime' to $\Sigma [{\mathbb Z}(m),{\mathbb Z}(m)]]$
in the sense that
\begin{eqnarray}
\Sigma [{\mathbb Z}(m),{\mathbb Z}(m)]\wedge [\Sigma [{\mathbb Z}(m),{\mathbb Z}(m)]\Rightarrow \Sigma [{\mathbb Z}(k),{\mathbb Z}(k)]]
\prec \Sigma [{\mathbb Z}(k),{\mathbb Z}(k)]
\end{eqnarray} 
Therefore $[\Sigma [{\mathbb Z}(m),{\mathbb Z}(m)]\Rightarrow \Sigma [{\mathbb Z}(k),{\mathbb Z}(k)]]$
contains `mainly' complementary information to $\Sigma [{\mathbb Z}(m),{\mathbb Z}(m)]$, but there is `some' overlap with it, which is
bounded by the information in $\Sigma [{\mathbb Z}(k),{\mathbb Z}(k)]$.
In analogy with Eqs.(\ref{M1}),(\ref{M2}), we show that
\begin{eqnarray}
(\neg \Sigma [{\mathbb Z}(m),{\mathbb Z}(m)])\vee \Sigma[{\mathbb Z}(k),{\mathbb Z}(k)]\prec 
\left (\Sigma [{\mathbb Z}(m),{\mathbb Z}(m)]\Rightarrow \Sigma [{\mathbb Z}(k),{\mathbb Z}(k)]\right ).
\end{eqnarray}
In a Boolean algebra these two quantities would have been equal to each other.

\paragraph*{Physical meaning of the equivalence:}
$[\Sigma [{\mathbb Z}(m),{\mathbb Z}(m)]\Leftrightarrow \Sigma [{\mathbb Z}(k),{\mathbb Z}(k)]]$
is the maximal system in ${\bf \Sigma}(n)$ which is `partly coprime' to both $\Sigma [{\mathbb Z}(m),{\mathbb Z}(m)]]$ and 
$ \Sigma [{\mathbb Z}(k),{\mathbb Z}(k)]$.

\subsection{Link of the logical operations to commuting von Neumann projectors}

We define the projector
\begin{eqnarray}\label{oper}
{\mathfrak P}(m)=\sum _{r=0}^{m-1}\ket {X_{k};rd}\bra{X_{k};rd};\;\;\;\;\;d=\frac{k}{m};\;\;\;\;m|k.
\end{eqnarray}
It might appear that we need to use an index $k$ and denote this projector ${\mathfrak P}_k(m)$, but
the isomorphism of Eq.(\ref{56}), which identifies the state $\ket {X_{m};r}$ in $H(m)$ with the state $\ket {X_{k};rd}$ in $H(k)$, 
in conjuction with the compatibility condition of Eq.(\ref{comp}),
imply that we can drop the index $k$.
Also, since we work in $\Lambda[{\bf \Sigma}(n)]$ and ${\Sigma}[{\mathbb Z}(n),{\mathbb Z}(n)]$ is the largest system,
${\mathfrak P}(n)={\bf 1}_n$.

Let ${\cal P} _k(m)$ (where $m|k$) be the projector into the 
$\varphi (m)$-dimensional subspace ${\cal A}_{mk}[h(m)]$ of $H(k)$.
From lemma \ref{790} follows that
\begin{eqnarray}\label{141}
{\cal P} _k(m){\cal P} _k(\ell)={\cal P} _k(m)\delta _{m\ell};\;\;\;\;\;\sum _{m|k}{\cal P} _k(m)={\mathfrak P}(k);\;\;\;\;
m|k;\;\;\;\;\;\ell|k.
\end{eqnarray}
A system ${\Sigma}[{\mathbb Z}(k),{\mathbb Z}(k)]$ in the state $\ket {s}$, `shares' the ${\cal P} _k(m)\ket{s}$ 
part of its state, with its subsystem
${\Sigma}[{\mathbb Z}(m),{\mathbb Z}(m)]$. The ${\cal P} _k(k)\ket{s}$ belongs only to 
${\Sigma}[{\mathbb Z}(k),{\mathbb Z}(k)]$ and it does not belong to its subsystems.

The state of a system can collapse into a state in one of its subsystems with appropriate measurements. 
We consider the system ${\Sigma}[{\mathbb Z}(n),{\mathbb Z}(n)]$ in a state described with the density matrix $\rho_n$. Then
\begin{eqnarray}\label{meas}
\tau (m|\rho _n)={\rm Tr}[\rho_n{\mathfrak P}(m)]
\end{eqnarray}
is the probability that a von Neumann measurement with the operator
\begin{eqnarray}
\Theta (m) =\theta _1{\mathfrak P}(m)+\theta _2[{\bf 1}_n-{\mathfrak P}(m)];\;\;\;\;m|n.
\end{eqnarray}
will collapse the system to the state 
\begin{eqnarray}
\rho =\frac{{\mathfrak P}(m)\rho_n{\mathfrak P}(m)}{\tau (m|\rho _n)}
\end{eqnarray}
of its subsystem ${\Sigma}[{\mathbb Z}(m),{\mathbb Z}(m)]$.
In a similar way we assign probabilities to the projectors ${\cal P}_n(m)$.
We note that $\tau(n|\rho _n)=1$.

The probabilities $\tau (m|\rho _n)$ are universal, in the sense of the following proposition.
\begin{proposition}\label{pro10}
If we embed the system ${\Sigma}[{\mathbb Z}(n),{\mathbb Z}(n)]$ into a larger system ${\Sigma}[{\mathbb Z}(u),{\mathbb Z}(u)]$ (with $n|u$),
then 
\begin{eqnarray}
\tau (m|\rho _n)=\tau [m|{\cal A}'_{nu}(\rho _n)];\;\;\;\;\;m|n|u
\end{eqnarray}
\end{proposition}
\begin{proof}
Using Eq.(\ref{oper}), we express the operator ${\mathfrak P}(m)$ as
\begin{eqnarray}
{\mathfrak P}(m)=\sum _{a=0}^{m-1}\ket {X_{n};ad_1}\bra{X_{n};ad_1}=\sum _{a=0}^{m-1}\ket {X_{u};ad_2}\bra{X_{u};ad_2};\;\;\;\;\;d_1=\frac{n}{m};\;\;\;\;d_2=\frac{u}{m}.
\end{eqnarray}
If 
\begin{eqnarray}
\rho _n=\sum _{a,b}\rho _n(a,b)\ket {X_{n};a}\bra{X_{n};b}
\end{eqnarray}
then
\begin{eqnarray}\label{eq10}
&&\rho _u={\cal A}'_{nu}(\rho _n)=\sum _{a,b}\rho _u(a',b')\ket {X_{u};a'}\bra{X_{u};b'}\nonumber\\
&&\rho _u(d_3a,d_3b)=\rho _n(a,b);\;\;\;\;\;\;d_3=\frac{u}{n}=\frac{d_2}{d_1}\nonumber\\
&&\rho_u(a',b')=0\;\;{\rm otherwise}
\end{eqnarray}
Therefore
\begin{eqnarray}
{\rm Tr}[\rho _u {\mathfrak P}(m)]=\sum _{a=0}^{m-1}\rho _u(ad_2,ad_2)=\sum _{a=0}^{m-1}\rho _n(ad_1,ad_1)={\rm Tr}[\rho _n {\mathfrak P}(m)].
\end{eqnarray}
\end{proof}

All logical operations are linked to von Neumann measurements.
For $m,k\in {\mathbb D}(n)$,
the ${\mathfrak P}(m\vee k)$ and ${\mathfrak P}(m\wedge k)$ are projectors to the spaces of 
the systems ${\Sigma}[{\mathbb Z}(m),{\mathbb Z}(m)]\vee {\Sigma}[{\mathbb Z}(k),{\mathbb Z}(k)]$
and ${\Sigma}[{\mathbb Z}(m),{\mathbb Z}(m)]\wedge {\Sigma}[{\mathbb Z}(k),{\mathbb Z}(k)]$, correspondingly (see Eq.(\ref{123})).
Starting from a state of ${\Sigma}[{\mathbb Z}(n),{\mathbb Z}(n)]$,
with these projectors we can get states in ${\Sigma}[{\mathbb Z}(m),{\mathbb Z}(m)]\vee {\Sigma}[{\mathbb Z}(k),{\mathbb Z}(k)]$
and ${\Sigma}[{\mathbb Z}(m),{\mathbb Z}(m)]\wedge {\Sigma}[{\mathbb Z}(k),{\mathbb Z}(k)]$.
In a similar way the ${\mathfrak P}(\neg m)$,  ${\mathfrak P}(m\Rightarrow k)$, ${\mathfrak P}(m\Leftrightarrow k)$ are projectors
to the spaces of the systems $\neg {\Sigma}[{\mathbb Z}(m),{\mathbb Z}(m)]$,
${\Sigma}[{\mathbb Z}(m),{\mathbb Z}(m)]\Rightarrow {\Sigma}[{\mathbb Z}(k),{\mathbb Z}(k)]$, 
${\Sigma}[{\mathbb Z}(m),{\mathbb Z}(m)]\Leftrightarrow {\Sigma}[{\mathbb Z}(k),{\mathbb Z}(k)]$, correspondingly.
All these projectors commute with each other.

For later use we also consider the probabilities
\begin{eqnarray}\label{meas}
{\widetilde \tau}(m|\rho _n)={\rm Tr}[\rho _n{\widetilde {\mathfrak P}}(m)]=\tau (m|\rho _n)-\tau (1|\rho _n);\;\;\;\;\;
{\widetilde {\mathfrak P}}(m)={\mathfrak P}(m)-{\mathfrak P}(1)
\end{eqnarray}
that the von Neumann measurement 
\begin{eqnarray}
{\widetilde \Theta} (m) =\theta _T{\widetilde {\mathfrak P}}(m)+\theta _F[{\bf 1}_n-{\widetilde {\mathfrak P}}(m)];\;\;\;\;m|n,
\end{eqnarray}
will collapse the system into a state in the space ${\widetilde H} (m)$ (defined in Eq.(\ref{cde})).
This measurement will give a `true-false' answer to whether the system will collapse to a state inside ${\widetilde H}(m)$.
Repetition of the experiment on an ensemble of systems in the same quantum state, will give the probability ${\widetilde \tau}(m|\rho _n)$.
We note that ${\widetilde \tau}(1|\rho _n)=0$.

\begin{remark}
We note that quantum logic is based on orthomodular lattices\cite{LO1,LO2,LO3,LO4}
(intuitionistic logic in this context is discussed in \cite{Co}).
They are the lattices of closed subspaces of Hilbert spaces, related to arbitrary von Neumann measurements, which in general are non-commutative.
The present work deals with a different problem and uses distributive lattices.
Of course, the space of a subsystem is a subspace of the space of the full system.
However, our concept `subsystem' contains the requirement that
the positions take values in a subgroup of the group of positions of the full system.
In this sense, the concept `subsystem' is intimately related to the `subgroup' rather than to the `subspace'.
A subsystem is a fundamental concept which retains its identity in the semiclassical limit.
A closed subspace is a secondary concept that is used in connection with measurements.
The lattice of the subgroups of a locally cyclic group is distributive\cite{LG2}. 
The groups for positions that we consider (${\mathbb Z}(n)$ and ${\mathbb Q}/{\mathbb Z}$) are locally cyclic 
and therefore it is not surprising that all our lattices are distributive.

The quantum theory describing the system $\Sigma ({\mathbb Z}(n), {\mathbb Z}(n))$ is of course non-commutative, but
the relationship between a system with its subsystems, and the logical connectives that describe it, are linked
to commutative projectors. This connects our quantities to commutative von Neumann measurements,
and this again explains the fact that in our context we use distributive lattices, rather than orthomodular lattices.
\end{remark}

\subsection{Probabilities for disjunctions and conjuctions and their incompatibility with associativity of the join in the Heyting algebra of subsystems $\Lambda[{\bf \Sigma}(n)]$}

\begin{proposition}
In ${\bf \Sigma }(n)$
\begin{itemize}
\item[(1)]
\begin{eqnarray}\label{61}
&&H(m\wedge k)=\bigoplus _{r|m\wedge k}{\cal A}_{r,m\wedge k}[h(r)]=H(m)\cap H(k)\nonumber\\
&&H(m\vee k)=\bigoplus _{r|m\vee k}{\cal A}_{r,m\vee k}[h(r)]
\end{eqnarray}
\item[(2)]
\begin{eqnarray}\label{t1}
&&{\mathfrak P}(m\wedge k)={\mathfrak P}(m){\mathfrak P}(k)\nonumber\\
&&{\mathfrak P}(m){\mathfrak P}(\neg m)={\mathfrak P}(1)
\end{eqnarray}
\item[(3)]
The map $m\;\rightarrow\;\tau (m|\rho _n)$ is order-preserving.
\end{itemize}
All variables in this proposition, belong to ${\mathbb D}(n)$.
\end{proposition}
\begin{proof}
\mbox{}
\begin{itemize}
\item[(1)]
Using lemma \ref{790}, we express the Hilbert spaces $H(k)$ and $H(m)$ as
\begin{eqnarray}\label{60}
H(m)=\bigoplus _{r|m}{\cal A}_{rm}[h(r)];\;\;\;\;H(k)=\bigoplus _{r|k}{\cal A}_{rk}[h(r)]
\end{eqnarray}
For the common divisors $r$ of $m$ and $k$, the $h(r)$ is isomorphic to both the ${\cal A}_{rm}[h(r)]$ and the ${\cal A}_{rk}[h(r)]$.
Physically, these three spaces contain the `same states' with different notations:
\begin{eqnarray}\label{560}
&&\sum _{r}a_r\ket{X_{m\wedge k};r}\;\;\rightarrow\;\;\sum _{r}a_r\ket{X_m;d_1r}\;\;\rightarrow\;\;\sum _{r}a_r\ket{X_k;d_2r}\nonumber\\
&&d_1=\frac{m}{m\wedge k};\;\;\;d_2=\frac{k}{m\wedge k};\;\;\;\;r=0,1,...,(m\wedge k)-1.
\end{eqnarray}
Therefore
\begin{eqnarray}\label{353}
H(m\wedge k)=\bigoplus _{r|m\wedge k}{\cal A}_{r,m\wedge k}[h(r)],
\end{eqnarray}
and similarly for $H(m\vee k)$.
The set of divisors of $m\wedge k$ is the intersection 
of the set of divisors of $m$ with the set of divisors of $k$.
Consequently
\begin{eqnarray}\label{616}
H(m\wedge k)=H(m)\cap H(k).
\end{eqnarray}
\item[(2)]
The first of these equations is proved using Eqs.(\ref{606}),(\ref{61}), in conjuction with the fact that the common divisors of $m,k$ are 
precisely the divisors of the $m\wedge k$.
The second equation follows immediately, from the first one.
\item[(3)]
If $m\prec k$ then ${\rm Tr}\{\rho _n{\mathfrak P}(k)[1-{\mathfrak P}(m)]\}\ge 0$ from which it follows that $\tau (m|\rho _n)\le \tau (k|\rho _n)$.
\end{itemize}
\end{proof}

\begin{proposition}

In ${\bf \Sigma} (n)$
\begin{itemize}
\item[(1)]
the space $T(m_1,m_2)$ 
\begin{eqnarray}\label{e1}
&&T(m_1,m_2)={\rm span}[H(m_1)\cup  H(m_2)]=\bigoplus _{r}{\cal A}_{r_1,m_1}[h(r)]\bigoplus _{r_2}{\cal A}_{r_2,m_2}[h(r)]\nonumber\\
&&r_1\in {\mathbb D}(m_1);\;\;\;\;\;\;r_2\in {\mathbb D}(m_1)-{\mathbb D}(m_1)\cap {\mathbb D}(m_2);\;\;\;\;\;s=m_1+m_2-m_1\wedge m_2,
\end{eqnarray}
is $s$-dimensional, it
contains all superpositions of states in $H(m_1)$ and $H(m_2)$, and it is a subspace of the space $H(m_1\vee m_2)$ of disjunctions. 
\item[(2)]
the  
\begin{eqnarray}\label{r1}
{\mathfrak T}(m_1,m_2)={\mathfrak P}(m_1)+{\mathfrak P}(m_2)-{\mathfrak P}(m_1\wedge m_2)
\end{eqnarray}
is projector to the space $T(m_1,m_2)$.
\end{itemize}
All variables in this proposition, belong to ${\mathbb D}(n)$.
\end{proposition}
\begin{proof}
\mbox{}
\begin{itemize}
\item[(1)]
The relation of Eq.(\ref{e1}) follows immediately from Eq.(\ref{606}), taking into account that we need to avoid double counting of the divisors. 
In order to prove that the dimension of the space $T(m_1,m_2)$ is equal to $s$, we use the equations
\begin{eqnarray}
&&\sum _{r_1}\varphi (r_1)=m_1;\;\;\;\;r_1\in {\mathbb D}(m_1)\nonumber\\
&&\sum _{r_2}\varphi (r_2)=m_2-m_1\wedge m_2;\;\;\;\;\;\;r_2\in {\mathbb D}(m_1)-{\mathbb D}(m_1)\cap {\mathbb D}(m_2)
\end{eqnarray}
The relation ${\mathbb D}(m_1)\cup {\mathbb D}(m_2)\subseteq {\mathbb D}(m_1\vee m_2)$
implies that the space $T(m_1,m_2)$ is a subspace of $H(m_1\vee m_2)$.
\item[(2)]
The proof of Eq.(\ref{r1}) is based on the relation $\sum _{r|m}{\cal P} _m(r)={\mathfrak P}(m)$ (Eq.(\ref{141})) 
in conjuction with the following identity from set theory:
\begin{eqnarray}
{\mathbb D}(m_1)\cup{\mathbb D}(m_2)=[{\mathbb D}(m_1)-{\mathbb D}(m_1)\cap {\mathbb D}(m_2)]\cup{\mathbb D}(m_2).
\end{eqnarray}
The term ${\mathfrak P}(m_1\wedge m_2)$ corrects the `double counting'.
\end{itemize}
\end{proof}
\begin{proposition}

In ${\bf \Sigma } (n)$
\begin{itemize}
\item[(1)]
the space
\begin{eqnarray}\label{e2}
&&S(m_1,m_2)=\bigoplus _{r}{\cal A}_{r,m_1\vee m_2}[h(r)]\nonumber\\
&&r\in {\mathbb D}(m_1\vee m_2)-[{\mathbb D}(m_1)\cup {\mathbb D}(m_2)]
\end{eqnarray}
is orthogonal to the space $T(m_1,m_2)$ and
\begin{eqnarray}\label{1000}
H(m_1\vee m_2)=T(m_1,m_2)\oplus S(m_1,m_2).
\end{eqnarray}
The $S(m_1,m_2)$ can be called `space of disjunctions which are not superpositions'.
\item[(2)]
The projector ${\mathfrak S}(m_1,m_2)$ to the space $S(m_1,m_2)$, is given by
\begin{eqnarray}\label{t2}
{\mathfrak S}(m_1,m_2)&=&{\mathfrak P}(m_1\vee m_2)-{\mathfrak T}(m_1,m_2)
\nonumber\\&=&{\mathfrak P}(m_1\vee m_2)-{\mathfrak P}(m_1)-{\mathfrak P}(m_2)
+{\mathfrak P}(m_1\wedge m_2).
\end{eqnarray}
The dimension of the space $S(m_1,m_2)$ is $m_1\vee m_2-m_1-m_2+m_1\wedge m_2$.

\item[(3)]
\begin{eqnarray}\label{467}
&&m_1\prec m_2\;\;\rightarrow\;\;{\mathfrak S}(m_1,m_2)=0\nonumber\\
&&{\mathfrak S}(m,n)={\mathfrak S}(1,m)=0\nonumber\\
&&r\prec m_1\;{\rm or}\;r\prec m_2\;\;\rightarrow \;\;{\mathfrak P}(r){\mathfrak S}(m_1,m_2)=0;
\end{eqnarray}
\item[(4)]
The quantum probabilities corresponding to the above projectors are
\begin{eqnarray}\label{t20}
\sigma (m_1,m_2|\rho _n)&=&{\rm Tr}[\rho _n{\mathfrak S}(m_1,m_2)]\nonumber\\
&=&\tau(m_1\vee m_2|\rho _n)-\tau (m_1|\rho _n)-\tau (m_2|\rho _n)
+\tau(m_1\wedge m_2|\rho _n)
\end{eqnarray}
The probabilities ${\widetilde \tau}(m|\rho _n)$ defined in Eq.(\ref{meas}), also satisfy this relation.
The $\sigma (m_1,m_2|\rho _n)$ quantify the difference between disjunctions and superpositions.
\item[(5)]
The above quantum probabilities, are compatible with the associativity of the join in the Heyting algebra of subsystems $\Lambda[{\bf \Sigma}(n)]$, 
only if $\sigma (m_1,m_2|\rho _n)=0$ for all density matrices, i.e., if the variables $m_i$ belong to the same chain within ${\mathbb D}(n)$.
\end{itemize}
All variables in this proposition, belong to ${\mathbb D}(n)$.
\end{proposition}
\begin{proof}
\item[(1)]
Lemma \ref{790}, and the fact that divisors entering in Eq.(\ref{e2}) for the space $S(m_1,m_2)$
are different from the divisors entering in Eq.(\ref{e1}) for the space $T(m_1,m_2)$, proves
the orthogonality of the two spaces.
Also the fact that
\begin{eqnarray}\label{e29}
\{{\mathbb D}(m_1\vee m_2)-[{\mathbb D}(m_1)\cup {\mathbb D}(m_2)]\}\cup [{\mathbb D}(m_1)\cup {\mathbb D}(m_2)]={\mathbb D}(m_1\vee m_2)
\end{eqnarray}
proves Eq.(\ref{1000}).
\item[(2)]
The orthogonality of the spaces $S(m_1,m_2)$ and $T(m_1,m_2)$, together with Eq.(\ref{r1}),
prove Eq.(\ref{t2}).
It then follows that the dimension of $S(m_1,m_2)$ is $m_1\vee m_2-m_1-m_2+m_1\wedge m_2$.

\item[(3)]
If $m_1\prec m_2$ then  ${\mathfrak P}(m_1\vee m_2)={\mathfrak P}(m_2)$ and 
${\mathfrak P}(m_1\wedge m_2)={\mathfrak P}(m_1)$ from which follows that
${\mathfrak S}(m_1,m_2)=0$. A corollary of this is that ${\mathfrak S}(m,n)={\mathfrak S}(1,m)=0$.

If $r\prec m_1$ then
\begin{eqnarray}
{\mathfrak S}(m_1,m_2){\mathfrak P}(r)&=&{\mathfrak P}[(m_1\vee m_2)\wedge r]-{\mathfrak P}(m_1\wedge r)-{\mathfrak P}(m_2\wedge r)
+{\mathfrak P}(m_1\wedge m_2\wedge r)\nonumber\\&=&{\mathfrak P}(r)-{\mathfrak P}(r)-{\mathfrak P}(m_2\wedge r)+{\mathfrak P}(m_2\wedge r)=0
\end{eqnarray}

\item[(4)]
This follows immediately from Eq.(\ref{t2}).
\item[(5)]
Classical probabilities $q(a_i)$ associated to events $a_i$, obey the relation
\begin{eqnarray}\label{cla}
q(a_1\vee a_2)-q(a_1)-q(a_2)+q(a_1\wedge a_2)=0.
\end{eqnarray}
This is proved in \cite{Cox1,Cox2,Jaynes}. The proof of our statement here, is similar to this but it involves quantum probabilities.

We first consider the case where any pair of the $m_1,m_2,m_3$ are coprime.
We assume that 
\begin{eqnarray}
{\widetilde \tau}(m_1 \vee m_2 |\rho _n)=F[{\widetilde \tau } (m_1|\rho _n),{\widetilde \tau }(m_2|\rho _n)]
\end{eqnarray}
where $F$ is a continuous function of two variables.
Then the associativity of the join operation gives
\begin{eqnarray}\label{997}
&&{\widetilde \tau }[(m_1\vee m_2)\vee m_3|\rho _n]={\widetilde \tau }[m_1\vee (m_2\vee m_3)|\rho _n]\;\;\rightarrow \nonumber\\
&&F\{F[{\widetilde \tau }(m_1|\rho _n),{\widetilde\tau }(m_2|\rho _n)],{\widetilde\tau }(m_3|\rho _n)\}=
F\{{\widetilde\tau }(m_1|\rho _n),F[{\widetilde \tau }(m_2|\rho _n),{\widetilde \tau }(m_3|\rho _n)]\}.
\end{eqnarray}
The equation
\begin{eqnarray}\label{998}
F[F(x,y),z]=F[x,F(y,z)]
\end{eqnarray}
is known as the associativity equation.
It is easily seen that if $F(x,y)$ can be written as 
\begin{eqnarray}\label{999}
F(x,y)=g^{-1}[g(x)+g(y)]
\end{eqnarray}
where $g(x)$ is any continuous strictly monotonic function, then Eq.(\ref{998}) is satisfied.
The converse also holds, but its
proof is complex for general continuous functions, and it is given in \cite{Aczel}.
The proof simplifies if we assume that $F$ is a differentiable function, and is given in \cite{Cox1,Cox2,Jaynes}.
Further discussion on the assumptions required for the proof can be found in \cite{H}.
From Eq.(\ref{999}) it follows that the associativity property of Eq.(\ref{997}), will be satisfied if there exist a function $g(x)$ 
such that for all density matrices
\begin{eqnarray}\label{1001}
g[{\widetilde \tau}(m_1 \vee m_2|\rho _n)]=g[{\widetilde \tau } (m_1|\rho _n)]+[g[{\widetilde \tau } (m_2|\rho _n)];\;\;\;\;\;m_1\wedge m_2=1.
\end{eqnarray}
But if there exists such a function it can only be $g(x)=\lambda x$,
because Eq.(\ref{t20}) shows that the density matrices for which $\sigma (m_1,m_2|\rho _n)=0$ with coprime $m_1,m_2$, obey the relation
\begin{eqnarray}\label{899}
{\widetilde \tau}(m_1 \vee m_2|\rho _n)={\widetilde \tau } (m_1|\rho _n)+{\widetilde \tau } (m_2|\rho _n).
\end{eqnarray}
Therefore we adopt the function $g(x)=\lambda x$, with $\lambda=1$.
Going from coprime $m_1,m_2$ to general $m_1,m_2$, we have to replace Eq.(\ref{1001}) with
\begin{eqnarray}
\tau(m_1\vee m_2|\rho _n)-\tau (m_1|\rho _n)-\tau (m_2|\rho _n)
+\tau(m_1\wedge m_2|\rho _n)=0. 
\end{eqnarray}
The term $\tau(m_1\wedge m_2|\rho _n)$ corrects the double counting.
This proves that the quantum probabilities, are compatible with the associativity of the join in the Heyting algebra $\Lambda[{\bf \Sigma}(n)]$, 
only if $\sigma (m_1,m_2|\rho _n)=0$.
This occurs for all density matrices, only if the variables $m_i$ belong to the same chain.
The term `probability' is intimately connected with certain properties, which are violated in the case of non-zero $\sigma (m_1,m_2|\rho _n)$.

\end{proof}
\begin{example}
We consider the systems in the set $\Sigma (18)$.
The $H(2\vee 3)=H(6)$. The space $H(2)$ (when embedded into $H(6)$) contains superpositions of the states
$\ket{X_6;0}$, $\ket{X_6;3}$, and the space $H(3)$ (when embedded into $H(6)$) contains superpositions of the states $\ket{X_6;0}$, $\ket{X_6;2}$, $\ket{X_6;4}$.
Therefore, the 
${\rm span}[H(2)\cup H(3)]$ is $4$-dimensional space and it contains superpositions of the states $\ket{X_6;0}$, $\ket{X_6;2}$, $\ket{X_6;3}$, $\ket{X_6;4}$.
Then the $S(2,3)$ is $2$-dimensional space and it contains superpositions of the states $\ket{X_6;1}$, $\ket{X_6;5}$. 
In terms of projectors
\begin{eqnarray}
&&{\widetilde {\mathfrak P}}(2)=\ket{X_6;3}\bra{X_6;3}\nonumber\\
&&{\widetilde {\mathfrak P}}(3)=\ket{X_6;2}\bra{X_6;2}+\ket{X_6;4}\bra{X_6;4}\nonumber\\
&&{\widetilde {\mathfrak P}}(2\vee 3)={\widetilde {\mathfrak P}}(2)+{\widetilde {\mathfrak P}}(3)+\ket{X_6;1}\bra{X_6;1}+\ket{X_6;5}\bra{X_6;5},
\end{eqnarray}
and from this follows that Eq.(\ref{899}) is not valid.
\end{example}
\begin{remark}
Classical probabilities are intimately related to set theory, and the probability of a disjunction is related to the union of sets.
The properties of sets lead to Eq.(\ref{cla}).
In Hilbert spaces we have two analogues to this, the $H(m_1\vee m_2)$ and the $T(m_1,m_2)={\rm span}[H(m_1)\cup H(m_2)]$.
Here $H(m_1\vee m_2)$ is the space of the smallest subsystem in the set $\Sigma (n)$, which contains the space ${\rm span}[H(m_1)\cup H(m_2)]$.
Consequently, we have two probabilities $\tau (m_1\vee m_2)$ and $\tau(m_1)+\tau(m_2)-\tau(m_1\wedge m_2)$
which are in general different, but they are the same when $m_1\prec m_2$, i.e., when $m_1,m_2$ are elements of a chain. 
This leads to the concept of contextuality, which here means that quantum probabilities are compatible 
with associativity of the join in the Heyting algebra of subsystems $\Lambda[{\bf \Sigma}(n)]$, only if the variables belong in the same chain (context).

The space $H(m_1\vee m_2)$ is invariant under the Fourier transforms $F_{m_1\vee m_2}$, while the space ${\rm span}[H(m_1)\cup H(m_2)]$ is not invariant under 
these Fourier transforms.
States in the space $S(m_1,m_2)$ belong entirely in the space $H(m_1\vee m_2)$ and yet they are not superpositions of states in $H(m_1)$ and $H(m_2)$. 
Disjunction of subsystems is more general concept than superposition.
This is because we require that there is a group of positions related to the states in $H(m_1\vee m_2)$,
which has as subgroups the groups ${\mathbb Z}(m_1)$ and ${\mathbb Z}(m_2)$ of positions, related to the states in the spaces $H(m_1)$ and $H(m_2)$,
correspondingly.
This is a feature of disjunction of quantum subsystems, which has no classical analogue. 
Usually non-classical behaviour is related to non-commutativity, superposition or entanglement. The disjunction of quantum subsystems as described above, 
is a novel feature which is not related to any of those.

\end{remark}

\subsection{Chains as contexts}

Various features of contextuality in quantum mechanics, have been discussed in the literature.
Ref \cite{C1} considers two sets of measurements $\{A,B,C,...\}$ and $\{A,L,M,...\}$. The measurements in each set 
commute with each other, but the measurements $B,C,...$, might not commute with the measurements $L,M,...$ (Mermin's square).
Contextuality is the fact that the measurement $A$ performed in conjuction with
$B,C,...$, might give a different result from the measurement $A$ performed in conjuction with the measurements $L,M,...$.
Another feature of contextuality is the lack of a joint probability distribution which has as marginals all the measured probability distributions.  The logical incompatibility of a set of measurements which violate Bell inequalities, has also been stressed in the literature.

In the present paper we have shown that associativity of the join in the Heyting algebra is incompatible with quantum probabilities,
which unlike classical probabilities, do not obey Eq.(\ref{cla}). However, within a chain quantum probabilities do obey this equation, and are compatible with the associativity of the join in the Heyting algebra. This leads to the following definition for a context.

\begin{definition}\label {def12}
A context is a subset of $\Lambda [\Sigma (n)]$ with labels in a subset ${\mathfrak C}(n)$ of ${\mathbb D}(n)$ 
which is a chain, so that $\sigma (m_1,...,m_\ell|\rho _n)=0$ for all $m_i\in {\mathfrak C}(n)$
and for all density matrices $\rho _n$. Then a relation analogous to Eq.(\ref{cla}) for classical probabilities, also holds for the quantum probabilities
${\widetilde \tau} (k|\rho _n)$, with $k\in {\mathfrak C}(n)$. An element of ${\mathbb D}(n)$ belongs in general, to many different contexts (chains).
\end{definition}

\begin{remark}
We can also define state-dependent contexts.
For a given set of density matrices $R=\{\rho _n\}$, a state-dependent context is a subset ${\mathfrak C}(n|R)$ of ${\mathbb D}(n)$, such that
$\sigma (m_1,...,m_\ell|\rho _n)=0$ for all $m_i\in {\mathfrak C}(n)$, and for all $\rho _n \in R$.
In this case a context can be larger than a chain.
In this paper we are interested in state-independent contexts, associated with chains, as in definition \ref{def12}.
\end{remark}
\begin{remark}
We use the term non-contextual quantum mechanics, for a theory where $\sigma (m_1,...,m_\ell|\rho _n)=0$ for all $m_i\in {\mathbb D}(n)$, and for all
density matrices. In such a theory, quantum probabilities obey Eq.(\ref{cla}), and therefore they are compatible with the associativity of the join in the Heyting algebra of subsystems. Below we show that logical Bell inequalities hold in non-contextual quantum mechanics.
They are violated by nature, and this proves that quantum mechanics is a contextual theory.
\end{remark}
\paragraph*{Pseudo-distances in non-contextual quantum mechanics}:
Eq.(\ref{t20}) with $\sigma (m_1,m_2|\rho _n)=0$ shows that ${\tau} (m|\rho _n)$ is a valuation.
In addition to that if $m\prec k$ then ${\tau} (m|\rho _n)\le {\tau} (k|\rho _n)$.
Therefore the
\begin{eqnarray}\label{bn}
{\mathfrak d}(m,k|\rho _n)={\widetilde \tau} (m\vee k|\rho _n)-{\widetilde \tau} (m\wedge k|\rho _n)
\end{eqnarray}
is a pseudo-distance, and $\Lambda[{\mathbb D}(n)]$ is a pseudo-metric lattice.
In contextual quantum mechanics, only a chain (context) $\Lambda[{\mathfrak C}(n)]$ is a pseudo-metric lattice.

\subsection{Logical Bell inequalities with Heyting factors }

Logical Bell inequalities have been studied in ref \cite{A1}, for the case of Boolean variables. 
In our logical Bell inequalities we use probabilities related with projectors to subsystems. Also, we have Heyting variables, and we get generalized logical Bell 
inequalities that contain `Heyting factors'.
\begin{proposition}\label{ineq}
In non-contextual quantum mechanics:
\begin{itemize}
\item[(1)]
Let  $m_1,...,m_\ell \in {\mathbb D}(n)-\{1\}$, and
$\rho _n$ be a density matrix describing the system ${\Sigma}[{\mathbb Z}(n),{\mathbb Z}(n)]$.
Then
\begin{eqnarray}\label{77}
\sum _{i=1}^\ell {\widetilde \tau}(m_i|\rho _n)\le \ell -\tau (r|\rho _n)-\sum _{i=1}^\ell f_i
\end{eqnarray}
where
\begin{eqnarray}\label{77a}
&&r=\neg (m_1\wedge ...\wedge m_\ell)=\prod _{p\in \pi}p^{e_p(n)};\;\;\;\pi=\varpi (n)-[\varpi (m_1)\cap...\cap \varpi (m_\ell)]\nonumber\\
&&f_i=1-\tau (m_i\vee \neg m_i|\rho _n)\ge 0.
\end{eqnarray}
$f_i$ are `Heyting factors', which are related to the difference between Heyting and Boolean algebras. 
If $m_i$ belongs to the Boolean algebra ${\mathbb D}^B(n)$, in which case $\Sigma [{\mathbb Z}(m_i),{\mathbb Z}(m_i)]$
is a $\pi$-Hall subsystem of $\Sigma [{\mathbb Z}(n),{\mathbb Z}(n)]$, then $f_i=0$.
\item[(2)]
In the special case that $m_1\wedge ...\wedge m_\ell=1$ this reduces to
\begin{eqnarray}\label{78}
\sum _{i=1}^\ell {\widetilde \tau}(m_i|\rho _n)\le \ell -1-\sum _{i=1}^\ell f_i
\end{eqnarray}
\item[(3)]
Eq.(\ref{77}) is universal, in the sense that if we
embed the system ${\Sigma}[{\mathbb Z}(n),{\mathbb Z}(n)]$ into a larger system ${\Sigma}[{\mathbb Z}(u),{\mathbb Z}(u)]$ (with $n|u$),
then the probabilities entering in Eq.(\ref{77}) remain the same: 
\begin{eqnarray}\label{m8}
&&{\widetilde \tau}(m|\rho _n)={\widetilde \tau}(m|\rho _u);\;\;\;\;\;\rho _u={\cal A}'_{nu}(\rho _n);\;\;\;\;m|n|u\nonumber\\
&&\tau (m\vee \neg m|\rho _n)=\tau (m\vee \neg 'm|\rho _u).
\end{eqnarray}
Here $\neg m$ and $\neg 'm$ are the negations of $m$ in ${\Lambda}[{\mathbb D}(n)]$ and ${\Lambda}[{\mathbb D}(u)]$), correspondingly.
\end{itemize}
\end{proposition}
\begin{proof}
\mbox{}
\begin{itemize}
\item[(1)]
We use the second de Morgan relation of Eq.(\ref{15A}), which is not valid in general Heyting algebras but 
it is valid in Stone lattices like $\Lambda [{\mathbb D}(n)]$. We get
\begin{eqnarray}\label{502}
r=\neg(m_1\wedge...\wedge m_\ell)=(\neg m_1)\vee...\vee (\neg m_\ell).
\end{eqnarray}
Therefore
\begin{eqnarray}
{\mathfrak P}[(\neg m_1)\vee...\vee (\neg m_\ell)]={\mathfrak P}(r).
\end{eqnarray}
For a density matrix $\rho _n$ of the system ${\Sigma}[{\mathbb Z}(n),{\mathbb Z}(n)]$, this leads to the probability
\begin{eqnarray}\label{7zx}
\tau [(\neg m_1)\vee...\vee (\neg m_\ell)|\rho _n]=\tau (r|\rho _n).
\end{eqnarray}
We next use Boole's inequality $q(a\vee b)\le q(a)+q(b)$. We note that in a contextual quantum mechanics,
we get
\begin{eqnarray}
\tau (k_1\vee k_2| \rho _n)\le \tau (k_1 | \rho _n)+\tau (k_2 | \rho _n)+\sigma (k_1,k_2|\rho _n), 
\end{eqnarray}
and Boole's inequality holds if $k_1,k_2$ belong in the same chain, so that $\sigma (k_1,k_2|\rho _n)=0$.
But in a non-contextual quantum mechanics, Boole's inequality holds and Eq.(\ref{7zx}) gives
\begin{eqnarray}\label{mo}
\tau (\neg m_1 | \rho _n)+...+\tau (\neg m_\ell| \rho _n)\ge \tau (r|\rho _n),
\end{eqnarray} 
We also use the relation
\begin{eqnarray}
\tau (m_i\vee \neg m_i | \rho _n)-\tau (\neg m_i | \rho _n)-\tau (m_i | \rho _n)+\tau (1| \rho _n)=0,
\end{eqnarray}
which holds in non-contextual quantum mechanics, to get
\begin{eqnarray}
\tau (m_i\vee \neg m_i | \rho _n)-{\widetilde \tau} (m_i | \rho _n)=\tau (\neg m_i | \rho _n).
\end{eqnarray}
Substitution in Eq.(\ref{mo}) proves the inequality in the proposition.
\item[(2)]
If $m_1\wedge ...\wedge m_\ell=1$, then $r=n$ and $\tau (n|\rho _n)=1$.
\item[(3)] 
The first of Eqs.(\ref{m8}) follows immediately from proposition \ref{pro10}.
The second equation is proved as follows:
\begin{eqnarray}
{\rm Tr}[\rho _n {\mathfrak P}(m \vee \neg m)]=\sum _{a=0}^{M_1-1}\rho _n(ad_1,ad_1);\;\;\;\;M_1=m \vee \neg m;\;\;\;\;\;d_1=\frac{n}{M_1}.
\end{eqnarray}
Taking into account Eq.(\ref{eq10}) we rewrite this as
\begin{eqnarray}
{\rm Tr}[\rho _n {\mathfrak P}(m \vee \neg m)]=\sum _{a=0}^{M_1-1}\rho _u(ag,ag);\;\;\;\;g=d_1\frac{u}{n}=\frac{u}{M_1}
\end{eqnarray}
We need to prove that this is equal to
\begin{eqnarray}\label{e34}
{\rm Tr}[\rho _u {\mathfrak P}(m \vee \neg 'm)]=\sum _{b=0}^{M_2-1}\rho _u(bd_2,bd_2);\;\;\;\;M_2=m \vee \neg 'm;\;\;\;\;\;d_2=\frac{u}{M_2}
\end{eqnarray}
But the fact that $n|u$, implies that $M_1|M_2$ and therefore $d_2|d_1|g$. Consequently, we
need to prove that in the sum of Eq.(\ref{e34}), $\rho _u(bd_2,bd_2)=0$ except for the cases that $bd_2$ is a multiple of $g$.
 
But according to Eq.(\ref{eq10}), $\rho _u(bd_2,bd_2)=0$ unless if $bd_2$ is an integer multiple of $u/n$:
\begin{eqnarray}
bd_2=c\frac{u}{n}\;\;\rightarrow \;\;b=c\frac{u}{nd_2}=\frac{c}{d_1}\frac{M_2}{M_1}
\end{eqnarray} 
In this equation, $c$ has to be a multiple of $d_1$, because $d_1$ and $M_2/M_1$ are coprime integers:
\begin{eqnarray}
&&\frac{M_2}{M_1}=\prod _{p\in \varpi (n)-\varpi(m)}p^{e_p(u)-e_p(n)}\prod _{p\in \varpi (u)-\varpi(n)}p^{e_p(u)};\;\;\;\;
\varpi(m)\subseteq \varpi(n)\subseteq \varpi(u)\nonumber\\
&&d_1=\prod _{p\in \varpi (m)}p^{e_p(n)-e_p(m)}
\end{eqnarray} 
Therefore $b$ is a multiple of $M_2/M_1$ and therefore $bd_2$ is a multiple of $g$.
This completes the proof.
\end{itemize}
\end{proof}
In the case $f_i=0$, Eq.(\ref{78}) can be understood as follows. Since $m_1\wedge...\wedge m_\ell=1$, a state (other than $\ket{X;0}$)
cannot belong to all subsystems $\Sigma (m_i)$, and therefore the sum of the $\ell$ probabilities ${\widetilde \tau}(m_i|\rho _n)$
cannot exceed $\ell -1$.

The following example shows that the logical Bell inequalities are violated, and therefore quantum mechanics is a contextual theory.

\begin{example}\label{105}
In order to get non-zero Heyting factors, $n$ needs to be $\prod p^{e_p}$ with some of the exponents $e_p\ge 2$.
We take $n=900$.
In ${\bf \Sigma}(900)$ we take $m_1=10$, $m_2=75$ and $m_3=36$.
We also consider the density matrix 
\begin{eqnarray}
\rho=a\ket {X_{900};180}\bra {X_{900};180}+b\ket {X_{900};25}\bra {X_{900};25}+(1-a-b)\ket {X_{900};5}\bra {X_{900};5}
\end{eqnarray}
where $a,b, (1-a-b)$ are probabilities.
Then the projectors from the system $\Sigma [{\mathbb Z}(900), {\mathbb Z}(900)]$ into the subsystems 
$\Sigma [{\mathbb Z}(10), {\mathbb Z}(10)]$, $\Sigma [{\mathbb Z}(75), {\mathbb Z}(75)]$
and $\Sigma [{\mathbb Z}(36), {\mathbb Z}(36)]$ are
\begin{eqnarray}
&&{\mathfrak P}(10)=\sum _{\nu=0}^9\ket {X_{900};90\nu}\bra {X_{900};90\nu};\;\;\;\;\;\;
{\mathfrak P}(75)=\sum _{\nu=0}^{74}\ket {X_{900};12\nu}\bra {X_{900};12\nu}\nonumber\\
&&{\mathfrak P}(36)=\sum _{\nu=0}^{35}\ket {X_{900};25\nu}\bra {X_{900};25\nu}.
\end{eqnarray}
Therefore
\begin{eqnarray}
{\widetilde \tau}(10;\rho )=a;\;\;\;\;\;{\widetilde \tau}(75;\rho )=a;\;\;\;\;\;\;
{\widetilde \tau}(36;\rho )=b.
\end{eqnarray}
Also 
\begin{eqnarray}
10\vee \neg 10=10\vee 9=90;\;\;\;\;\;75\vee \neg 75=75\vee 4=300;\;\;\;\;\;36\vee \neg 36=36\vee 25=900
\end{eqnarray}
and
\begin{eqnarray}
&&{\mathfrak P}(10\vee \neg 10)=\sum _{\nu=0}^{89}\ket {X_{900};10\nu}\bra {X_{900};10\nu};\;\;\;\;\;\;
{\mathfrak P}(75\vee \neg 75)=\sum _{\nu=0}^{299}\ket {X_{900};3\nu}\bra {X_{900};3\nu}\nonumber\\
&&{\mathfrak P}(36\vee \neg 36)={\mathfrak P}(900)={\bf 1}
\end{eqnarray}
Therefore
\begin{eqnarray}
\tau (10\vee \neg 10;\rho )=a;\;\;\;\;\;\;\tau (75\vee \neg 75;\rho )=a;\;\;\;\;\;\;
\tau (36\vee \neg 36;\rho )=1
\end{eqnarray}
and the Heyting factors are $f_1=1-a$, $f_2=1-a$ and $f_3=0$. 
Therefore the inequality of Eq.(\ref{78}) becomes $a+a+b\le 3-1-(1-a)-(1-a)$ and is violated.

We note that $m_3$ belongs to the Boolean algebra ${\mathbb D}^B(900)$, and $f_3=0$, as stated in the proposition.
\end{example}

\section{The complete Heyting algebra $\Lambda ({\bf \Sigma} _S)$ of quantum systems}\label{6}

The infinite set $\{\Sigma ({\mathbb Z}(n), {\mathbb Z}(n))\;|\;n\in {\mathbb N}\}$,
is a distributive lattice but it is not complete. In order to make it complete,
we enlarge it as follows:
\begin{eqnarray}
{\bf \Sigma} _S=\{\Sigma [{\cal C}(n),\widetilde {\cal C}(n)]\;|\;n\in {\mathbb N}_S\}.
\end{eqnarray}
This set contains the quantum system $\Sigma [{\cal C}(p^ {\infty}), \widetilde {\cal C}(p^ {\infty})]=\Sigma ({\mathbb Q}_p/{\mathbb Z}_p,{\mathbb Z}_p)$
where the position takes values in ${\mathbb Q}_p/{\mathbb Z}_p$ and the momentum takes values in
${\mathbb Z}_p$, which has been studied as a subject in its own right 
in refs \cite{VOU2,VOU3}. It also contains the 
quantum system $\Sigma [{\cal C}(\Omega), \widetilde {\cal C}(\Omega)]=\Sigma ({\mathbb Q}/{\mathbb Z},\widehat {\mathbb Z})$,
where the position takes values in ${\mathbb Q}/{\mathbb Z}$ and the momentum takes values in $\widehat {\mathbb Z}$, which
has been studied as a subject in its own right in refs \cite{VOU4}. 
Table 1 presents a summary of these systems. 
Here we extend the formalism of the previous section to the set ${\bf \Sigma} _S$.
In particular, we define quantities analogous to those used in the previous section, so that the formalism developed there, is also applicable here.

It is easily seen that ${\bf \Sigma} _S$ is isomorphic to ${\mathfrak Z} _S$ and $\widetilde {{\mathfrak Z} _S}$,
and therefore it is a complete Heyting algebra with
\begin{eqnarray}
{\cal O}=\Sigma ({\mathbb Z}(1),{\mathbb Z}(1));\;\;\;\;\;{\cal I}=\Sigma ({\mathbb Q}/{\mathbb Z},\widehat {\mathbb Z}). 
\end{eqnarray}
This is the analogue of Eqs.(\ref{bt}),(\ref{300}) for the relevant groups. We denote this algebra $\Lambda ({\bf \Sigma} _S)$.

From Eqs(\ref{D1}),(\ref{D2}), we see that if $\Sigma [{\cal C}(n),\widetilde {\cal C}(n)]$ is a $\pi$-system then
\begin{eqnarray}
\neg \Sigma [{\cal C}(n),\widetilde {\cal C}(n)]=\Sigma \left [\bigoplus _{p\in \Pi-\pi}({\mathbb Q}_p/{\mathbb Z}_p),\prod _{p\in \Pi-\pi}{\mathbb Z}_p
\right ] 
\end{eqnarray}
The implication can be found using Eqs.(\ref{100B}),(\ref{100BB}).

The subset of ${\bf \Sigma} _S$ given by
\begin{eqnarray}
{\bf \Sigma} _S^B=\{\Sigma ({\cal C}[\Omega (\pi)],\widetilde {\cal C}[\Omega (\pi)])\;|\;\pi \subseteq \Pi\},
\end{eqnarray}
with the $\wedge$ and $\vee$ operations, is a Boolean algebra. 
This is the analogue of Eqs.(\ref{129}), (\ref{301}) for the relevant groups, and it contains $\pi$-Hall subsystems of $\Sigma ({\mathbb Q}/{\mathbb Z},\widehat {\mathbb Z})$.

In Eq.(\ref{oper}) we have defined the projectors ${\mathfrak P}(n)$ with $n\in {\mathbb N}$.
Here we extend this definition, and define
the projector ${\mathfrak P}(n)$ (with $n\in {\mathbb N}_S$) into the Schwartz-Bruhat space of the system
$\Sigma [{\cal C}(n), \widetilde {\cal C}(n)]$.
The Schwartz-Bruhat space ${\mathfrak B}({\mathbb Q}/{\mathbb Z},\widehat {\mathbb Z})$ of the system 
$\Sigma ({\mathbb Q}/{\mathbb Z},\widehat {\mathbb Z})$ has been defined in \cite{VOUREV} (definition 6.1),
and the Schwartz-Bruhat space ${\mathfrak B}[{\cal C}(n), \widetilde {\cal C}(n)]$ of the system 
$\Sigma [{\cal C}(n), \widetilde {\cal C}(n)]$ has been defined in \cite{VOUREV} (definition 7.1).
The projector ${\mathfrak P}(n)$ (with $n\in {\mathbb N}_S$) maps the function
$f(x) \in {\mathfrak B}({\mathbb Q}/{\mathbb Z},\widehat {\mathbb Z})$ where $x\in {\mathbb Q}/{\mathbb Z}$
into the function 
\begin{eqnarray}
&&[{\mathfrak P}(n)f](x)=f(x)\;\;{\rm if}\;x\in {\cal C}(n)\nonumber\\
&&[{\mathfrak P}(n)f](x)=0\;\;{\rm if}\;x\notin {\cal C}(n)
\end{eqnarray}
Then the probabilities $\tau [n|f(x)]$ are given by 
\begin{eqnarray}
\tau [n|f(x)]=\int _{{\mathbb Q}/{\mathbb Z}}|[{\mathfrak P}(n)f](x)|^2dx
\end{eqnarray}
For functions in the Schwartz-Bruhat space these integrals are finite sums, and they converge.

\section{Discussion}

We have developed a lattice theory language for the set $\{G_i\}$ of subgroups of an Abelian group $G$
and also for the set $\{{\widetilde G}_i\}$ of the Pontryagin duals of these groups.
For $G={\mathbb Q}/{\mathbb Z}$ we get the Heyting algebra $\Lambda ({\mathfrak Z}_S)$, and for
${\widetilde G}={\widehat {\mathbb Z}}$ we get the Heyting algebra $\Lambda ({\widetilde {\mathfrak Z}_S})$.
We have discussed the meaning of the logical connectives in this formalism.

We have considered quantum systems $\Sigma (G_i, {\widetilde G}_i)$, with positions in $G_i$ and momenta in
${\widetilde G}_i$. We have studied the finite set ${\bf \Sigma}(n)$ of subsystems of the   
$\Sigma [{\mathbb Z}(n),{\mathbb Z}(n)]$ as a finite Heyting algebra, and we have discussed the physical meaning of the logical connectives 
and of the non-validity of the law of the exclusive middle, in this formalism.
Ideas from Sylow theory for finite groups, have been transfered into the corresponding quantum systems.
 
We have shown that quantum probabilities, related to projectors in the subsystems,  are incompatible
with associativity of the join in the Heyting algebra.
This is because of the probabilities $\sigma(m_1,m_2|\rho _n)$, which are
related to the fact that the `disjunction space' $H(m_1 \vee m_2)$ is larger than the space ${\rm span}[H(m_1)\cup H(m_2)]$ of superpositions.
Disjunction is more general concept than superposition.
This leads to contextuality, which in the present formalism has as contexts, chains of the Heyting algebra.
Within a chain, Eq.(\ref{cla})
for the classical probabilities of disjunctions, is also valid for quantum probabilities. 
If quantum mechanics were a non-contextual theory, it would obey the logical Bell inequalities
of Eq.(\ref{77}), which involve projectors to subsystems, and which
generalize previous logical Bell inequalities\cite{A1}, for Heyting (as opposed to Boolean) variables.
In example \ref{105} we have shown that quantum mechanics violates these inequalities, and therefore it is a contextual theory.

The infinite set $\{\Sigma [{\mathbb Z}(n),{\mathbb Z}(n)]\;|\;n\in {\mathbb N}\}$ 
is a directed partially ordered set with the partial order `subsystem'.
It is not a directed-complete partial order, and this could be interpreted as `something is missing'.
We have added the `top elements' in this set (the systems in the last three rows in table 1) 
in order to make it a directed-complete partial order\cite{VOU1,VOUREV}. 
This enlarged set is a complete Heyting algebra, and in section \ref{6} we have defined, 
the various quantities used earlier, for these systems also.  

Unlike the logic of quantum measurements which is based on orthomodular lattices, the logic of subsystems studied here, uses distributive lattices (Heyting algebras).
In this case the disjunction is not equivalent to superposition, and this provides a different insight to contextuality.

\newpage
\begin{table}
\caption{The groups $C(n)$ where $n\in {\mathbb N}_S$, their Pontryagin dual groups $\widetilde C(n)$, the quantum systems
$\Sigma [{\cal C}(n),\widetilde {\cal C}(n)]$, and the corresponding projectors ${\mathfrak P}(n)$}
\centering
\begin{tabular}{|c|c|c|c|}\hline
${\cal C}(n)$& $\widetilde {\cal C}(n)$& $\Sigma [{\cal C}(n),\widetilde {\cal C}(n)]$& ${\mathfrak P}(n)$\\ \hline\hline
${\cal C}(n)\cong {\mathbb Z}(n)$ ($n\in {\mathbb N}$) & $\widetilde {\cal C}(n)\cong {\mathbb Z}(n)$ &$\Sigma [{\mathbb Z}(n),{\mathbb Z}(n)]$&${\mathfrak P}(n)$\\ \hline
${\cal C}(p^\infty)\cong {\mathbb Q}_p/{\mathbb Z}_p$ & $\widetilde {\cal C}(p^\infty)\cong{\mathbb Z}_p$&$\Sigma [{\mathbb Q}_p/{\mathbb Z}_p,{\mathbb Z}_p]$&${\mathfrak P}(p^\infty)$\\ \hline
${\cal C}[(\Omega (\pi)]\cong  \bigoplus _{p\in \pi}{\mathbb Q}_p/{\mathbb Z}_p$ & $\widetilde {\cal C}[(\Omega (\pi)]\cong \prod _{p\in \pi}{\mathbb Z}_p$&
$\Sigma [\prod _{p\in \pi}{\mathbb Q}_p/{\mathbb Z}_p,\prod _{p\in \pi}{\mathbb Z}_p]$&${\mathfrak P}(\Omega (\pi))$\\ \hline
${\cal C}(\Omega )\cong {\mathbb Q}/{\mathbb Z}\cong \bigoplus _{p\in \Pi}{\mathbb Q}_p/{\mathbb Z}_p$ & $\widetilde {\cal C}(\Omega )\cong {\widehat {\mathbb Z}}\cong \prod _{p\in \Pi}{\mathbb Z}_p$
&$\Sigma ({\mathbb Q}/{\mathbb Z},{\widehat {\mathbb Z}})$&${\bf 1}$\\ \hline
\end{tabular}
\end{table}

\end{document}